%% file: Fingerprints.tex
\documentclass[a4paper,envcountsame]{llncs}

\usepackage[UKenglish]{babel}

\usepackage{Fingerprints}
\usepackage{microtype}

\title{The model checking fingerprints\\ of CTL operators}

\author{Andreas Krebs\inst{1} \and Arne Meier\inst{2} \and Martin Mundhenk\inst{3}}
\institute{Universität Tübingen, Sand 13, 72076 Tübingen, Germany, \email{krebs@informatik.uni-tuebingen.de} \and
  Leibniz Universität Hannover, Appelstra\ss{}e 4, 30167 Hannover, Germany, \email{meier@thi.uni-hannover.de} \and
  Friedrich-Schiller-Universit\"at Jena, Ernst-Abbe-Platz 2, 07743 Jena, Germany, \texttt{martin.mundhenk@uni-jena.de}
}

\begin{document}

\maketitle

\begin{abstract}
 The aim of this study is to understand the inherent expressive power of \CTL operators. 
 We investigate the complexity of model checking for all \CTL fragments
 with one \CTL operator and arbitrary Boolean operators.
 This gives us a fingerprint of each \CTL operator.
 The comparison between the fingerprints 
 yields a hierarchy of the operators that mirrors their strength with respect to model checking.
\end{abstract}

\input{Introduction}

\input{Preliminaries}


\section{Computation Tree Logic CTL}

\input{EU}

\input{remainingOperators}

%



%
%

\input{conclusion}

\bibliographystyle{plain}
\bibliography{mc-jena}

\section{Appendix}

\input{Appendix-EU}

\input{Appendix-ER}

\input{Appendix-EG}

\input{Appendix-EF}

\end{document}

%% file: Introduction.tex
\section{Introduction}
Temporal logics are a long used and well-understood concept to model software specifications and computer programs by state transition semantics. 
The first approaches in this currently quite large area of research go back to Arthur N.\ Prior \cite{pr57,pr67}. 
The logics became more prominent in the 70s and 80s due to significant effort of Pnueli, Emerson, Halpern, and Clarke \cite{pnu77,clem81,emha86}. 
Usually one distinguishes between three temporal logics: linear time logic \LTL, computation tree logic \CTL, and the full branching time logic \CTLstar. 
All these logics are defined as extensions of (modal) propositional logic to express properties of computer programs 
by introducing two path quantifiers $\A$ and $\E$, resp., five temporal operators ne$\X$t, $\U$ntil, $\F$uture, $\G$lobally, and $\R$elease. 
Form a syntatctic point of view,
the three temporal logics differ in the way how the path quantifiers and temporal operators may be combined.  
The computation tree logic \CTL allows operators that are combined from one path quantifier directly followed by one temporal operator.
Thus there are ten different \CTL operators---e.g., $\EX$ or $\AU$. 

The most important decision problems related to temporal logics are the satisfiability problem and the model checking problem. 
The complexity of these problems ranges between $\P$ and $\TwoEXPTIME$ and 
has been classified for the general cases \cite{vast85,Var85a,emju00,fila79,pr80,clemsi86,sch02}. 
Recently the satisfiability problem for all three logics has been completely classified 
with respect to all Boolean and temporal operator fragments \cite{mmtv09,bsssv09},
motivated in part by the fundamental work of E.\ Post \cite{pos41} on Boolean functions.
In the same way, the model checking problem for $\LTL$  was studied in detail \cite{BaulandLTL11}. 
The model checking problem for \CTL has been deeper understood in \cite{Beyer11} 
who examined the complexity of \CTL fragments 
that have arbitrary \CTL operators that are combined only with all monotone Boolean operators.
For model checking, there are seven relevant fragments of Boolean operators \cite{BaulandLTL11},
but only one of these was considered in \cite{Beyer11}.

\begin{figure}[t]
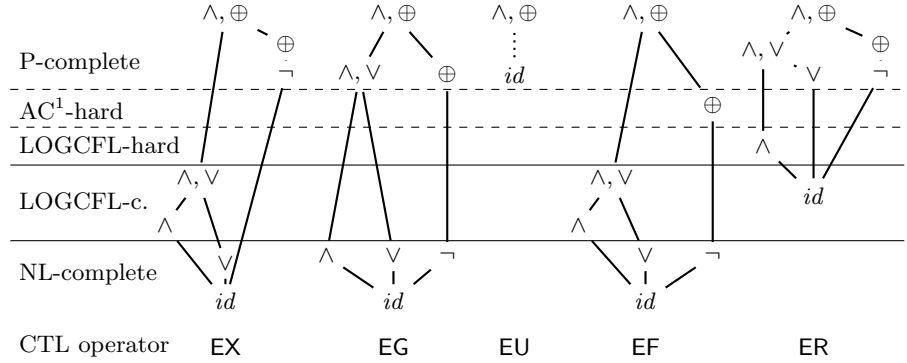

\centering
\CTLResultate
\caption{Overview of complexity results---the model checking fingerprints of the \CTL operators.}
\label{fig:complexity-overview}
\end{figure}

\begin{wrapfigure}{r}{.5\textwidth}
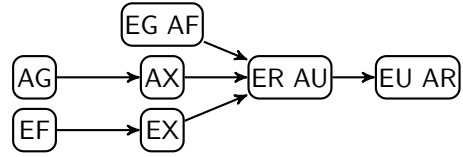

\centering
  \CTLOperatorHierarchie

\caption{The mc-strength hierarchy of \CTL operators that relies on their fingerprints (Fig.~\ref{fig:complexity-overview}).
Arrows indicate the relation $\morepowerfull$.
Operators $X$ and $Y$ in the same circle have the same mc-strength (i.e., $X\morepowerfull Y$ and $Y\morepowerfull X$).
The hierarchy is proper under the common assumptions $\NL\subsetneq \LOGCFL \subsetneq \P$.
}\label{fig:OperatorHierarchie}
\vspace{-15pt}
\end{wrapfigure}
We aim to fill this gap by classifying the remaining relevant Boolean operator fragments for the computation tree logic \CTL.
More specifically,
we examine the complexity of \CTL model checking 
for all fragments of formulas that combine one of the ten \CTL operators with one of the seven relevant fragments of Boolean operators.  
With our work one can completely characterize all but four of these combinations.
Our classifications---informally called fingerprints---yield a preorder expressing how powerful a \CTL operator is.
We say a \CTL operator $\mathcal T$ is \emph{mc-stronger} than $\mathcal T'$, in symbols $\mathcal T'\morepowerfull\mathcal T$, if for every set $B$ of Boolean operators the model checking problem for the $(\{\mathcal T\}\cup B)$-fragment of \CTL is computationally harder than that for the $(\{\mathcal T'\}\cup B)$-fragment. The resulting partial order is shown in Figure~\ref{fig:OperatorHierarchie}.
It can be seen as a generalization of the notion of expressiveness of \CTL fragments \cite{laroussinie95}.
Whereas the notion of expressiveness deals with equivalence of formulas from different fragments, our notion of mc-strength deals with equivalence of model checking instances for different fragments.
Expressiveness is meaningful from a language theoretic point of view, and mc-strength from the computational complexity perspective.

The paper is organized as follows. 
At first we introduce syntax and semantics of \CTL, and we explain the alternating graph accessibility problems that we use in our hardness proofs (Section~\ref{sec:prelims}).
We visit each \CTL operator and show its complexity fingerprint (Sections~\ref{sec:EU} and \ref{sec:remaining}).
Finally we conclude with the resulting comparison of the mc-strength of \CTL operators and an outlook to future work (Section~\ref{sec:conclusion}).
The Appendix contains the missing proofs.

%% file: Preliminaries.tex
\section{Preliminaries}
\label{sec:prelims}

\subsection{Computation Tree Logic CTL}

Let $\PROP$ be a set of atomic propositions. Then the set of all well-formed \CTL formulas is
$
\varphi ::= 
\true\mid
p\mid
\varphi\land\varphi\mid
\varphi\lor\varphi\mid
\varphi\xor\varphi\mid
\lnot\varphi\mid
\mathcal{P}\mathcal{O}\varphi\mid
\varphi\mathcal{P}\mathcal{O}'\varphi, 
$
for $p\in\PROP$, $\mathcal{P}\in\{\A,\E\}$, $\mathcal{O}\in\{\X,\F,\G\}$, $\mathcal{O}'\in\{\U,\R\}$. 
We say $\mathcal{P}\mathcal{O}$ are the unary \CTL operators $\EX$, $\AX$, $\EG$, $\AG$, $\EF$, $\AF$ 
and $\mathcal P\mathcal O'$ are the binary \CTL operators $\EU$, $\AU$, $\ER$, $\AR$.
A Kripke model (for \CTL) is a triple $(W,R,\xi)$, where $W$ is a finite set of states, 
$R\colon W\to W$ is a total transition relation (i.e., for all $w\in W$ there is a $w'\in W$ with $wRw'$), 
and $\xi\colon W\to 2^{\PROP}$ is an assignment function.

The semantics of \CTL is defined as follows on states.
Let $\Model M=(W,R,\xi)$ be a Kripke model. 
Let $\Pi(w)$ denote the set of infinite paths starting in $w\in W$ through $(W,R)$, 
i.e., a path $\pi\in\Pi(w)$ is an infinite sequence $\pi=\pi[1]\pi[2]\cdots$ with $\pi[1]=w$ and $(\pi[i],\pi[i+1])\in R$ for all $i\geq 1$.
%
{
$$
\begin{array}{lcl}
 \Model M, w\models \true & \multicolumn{2}{l}{\text{always,}}\\
 \Model M, w\models p & \text{iff} & p\in\xi(w),\\
 \Model M, w\models \lnot\psi & \text{iff} &  \Model M, w\not\models\psi,\\
 \Model M, w\models \psi\land\phi & \text{iff} & \Model M, w\models\psi \text{ and }\Model M, w\models\phi,\\
 \Model M, w\models \psi\lor\phi & \text{iff} & \Model M, w\models\psi \text{ or }\Model M, w\models\phi,\\
 \Model M, w\models \psi\oplus\phi & \text{iff} & (\Model M, w\models\psi \text{ and }\Model M, w\nmodels\phi) \text{~or~} (\Model M, w\nmodels\psi \text{ and }\Model M, w\models\phi),\\
 \Model M, w \models \EX\varphi & \text{iff} & \exists \pi\in\Pi(w): ~\Model  M, \pi[2]\models \varphi,\\
 \Model M, w \models \EF\varphi & \text{iff} & \exists\pi\in\Pi(w) ~\exists k\ge 1: ~\Model M, \pi[k]\models \varphi,\\
 \Model M, w \models \EG\varphi & \text{iff} & \exists\pi\in\Pi(w) ~\forall k\ge 1: ~\Model M, \pi[k]\models \varphi,\\
 \Model M, w \models \psi\EU\varphi & \text{iff} & \exists\pi\in\Pi(w) ~\exists k\ge 1: ~\Model M, \pi[k]\models \varphi\text{ and } \forall i<k: \Model M,\pi[i]\models\psi,\\
 \Model M, w \models \psi\ER\varphi & \text{iff} & \exists\pi\in\Pi(w) ~\forall k\ge 1: ~\Model M, \pi[k]\models \varphi\text{ or } \exists i< k: \Model M,\pi[i]\models\psi.\\
\end{array}
$$
}

%

The remaining \CTL operators can be expressed as duals of the above defined operators.
We have the equivalences $\AX\varphi\equiv\lnot\EX\lnot\varphi, \AF\varphi\equiv\lnot\EG\lnot\varphi$, $\AG\varphi\equiv\lnot\EF\lnot\varphi$,
$\psi\AR\varphi \equiv \neg(\neg\psi\EU\neg\varphi)$, and $\psi\AU\varphi \equiv \neg(\neg\psi\ER\neg\varphi)$.
Moreover, the operators $\EX,\EG,\EU$ are a minimal set of CTL operators that together with the Boolean operators 
suffice to express any from the others \cite{laroussinie95}, and with the Boolean operators $\land,\oplus$ one can express every Boolean function.
For a set $T\subseteq \{\EX$, $\AX$, $\EG$, $\AG$, $\EF$, $\AF$, $\EU$, $\AU$, $\ER$, $\AR$, $\land$, $\lor$, $\neg$, $\oplus\}$ of Boolean functions and \CTL operators, a $T$-formula is a formula that has operators only from $T$.
The $T$-fragment of \CTL is the set of all $T$-formulas.
The model checking problems for \CTL fragments are defined as follows. 

\problemdef{$\CTLMC T$}{The model checking problem for $T$-fragments of \CTL}%
  { A \CTL formula $\phi$ with operators in $T\subseteq\{\EX$, $\AX$, $\EG$, $\AG$, $\EF$, $\AF$, $\EU$, $\AU$, $\ER$, $\AR$, $\land$, $\lor$, $\neg$, $\oplus\}$,
    a Kripke model $\Model{M}=(W,R,\xi)$, and a state $w_0\in W$%
  }{%
    Does $\Model{M},w_0 \models \phi$ hold%
  }
Usually we will omit the $\{\cdot\}$ and $\cup$ in the problem notion for convenience.

%
  Post \cite{pos41} classified the lattice of all relevant sets of Boolean operators---called \emph{clones}---and 
  found a finite base for each clone.
  The definitions of all clones as well as the full inclusion graph can be found, for example, in~\cite{bcrv03}.
  Whereas in general there is an infinite set of clones,
  for model checking luckily there are only seven different clones~\cite{BaulandLTL11}
  depicted in Figure~\ref{fig:mc-clones}, where we describe the clones by their standard bases. 
  (See, e.g., \cite{mc_survey_beatcs} for more explanations.)

\begin{wrapfigure}{r}{.3\textwidth}
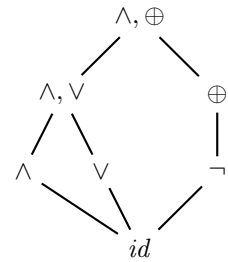

\centering
  \ClonesBild

 \caption{The Boolean clones relevant for model checking, represented by their standard bases.
\emph{id} denotes the clone represented without operator (``identity'' of an atom).}
\label{fig:mc-clones}
\end{wrapfigure}

\subsection{Computational Complexity}
\label{subsubsec:Kripke}

We will make use of standard notions of complexity theory \cite{pap94}. 
In particular, we will make use of the complexity classes $\NL,\LOGCFL,\AC1,$ and $\P$. 

$\NL$ is the class of problems decided by nondeterministic logarithmically space bounded Turing machines.
The typical complete problem is the graph accessibility problem for directed graphs $\REACH$ (given a directed graph
with two nodes $s$ and $t$, is there a path form $s$ to $t$?).
$\LOGCFL$ is the class of problems decided by nondeterministic logarithmically space bounded Turing machines,
that are additionally allowed to use a stack and run in polynomial time.
$\AC1$ is the class of problems decided by alternating logarithmically space bounded Turing machines
with logarithmically bounded number of alternations.
We will shortly present complete problems for both of these classes.
In order to prove hardness results, we will make use of logarithmic space bounded many-one reductions $\leqlogm$.
It is known that $\NL\subseteq\LOGCFL\subseteq\AC1\subseteq\P$ but not whether any inclusion is strict.

Clarke, Emerson, and Sistla~\cite{clemsi86} 
showed that model checking for \CTL is in $\P$,
and Schnoebelen \cite{sch02} showed that it is $\P$-hard.

\begin{theorem}[\!\!{\cite{clemsi86,sch02}}]\label{thm:CTLMC}
$\CTLMC{\EX,\ldots,\AR,\wedge,\vee,\neg,\oplus}$ is $\P$-complete.
\end{theorem}

How \CTL operators compare with respect to the complexity of model checking,
was investigated in~\cite{Beyer11} in the following way.
They completely characterize the complexity of $\CTLMC{T,\land,\lor}$
for every set $T$ of \CTL operators.
They show that this complexity is either $\P$-complete or $\LOGCFL$-complete.
For singletons $S\subset\{\AF,\EG,\AU,\EU,\AR,\ER\}$ the problems $\CTLMC{S,\land,\lor}$ are $\P$-complete,
whereas for all other singletons $S\subset\{\AX,\EX,\EF,\AG\}$ $\CTLMC{S,\land,\lor}$ is only $\LOGCFL$-complete.

\input{AsAgap}

%% file: AsAgap.tex
\begin{wrapfigure}{r}{.5\textwidth}
\vspace{-.85cm}
\centering
\ASAGAPveeonein

\caption{An instance $\langle G,s,T\rangle$ of $\ASAGAPtwooutonein$.
The marked edges indicate the witness for $\apath_G(s,T)$.}
\label{fig:Bsp-ASAGAPveeonein}
\end{wrapfigure}
Next, we consider problems that we will use
for reductions in our hardness proofs.
The alternating graph accessibility problem is shown to be $\P$-complete in \cite{chkost81}.
We use the following restricted version of this problem 
that is very similar to Boolean circuits with and- and or-gates (and input-gates).
An \emph{alternating slice graph} \cite{MW13} $G=(V,E)$ is
a directed bipartite acyclic graph with a bipartitioning $V=V_{\exists} \cup V_{\forall}$,
and a further partitioning
$V  =  V_0 \cup V_1 \cup \cdots \cup V_{m}$ ($m+1$ \emph{slices}, $V_i\cap V_j=\emptyset$ if $i\not=j$)
where
$$
	V_{\exists}  = \bigcup\limits_{i\leq m, i \text{ even}} V_i \text{ and } V_{\forall}=\bigcup\limits_{i\leq m, i \text{ odd}} V_i,
\text{ such that } 
	E  \subseteq \bigcup\limits_{i=0}^{m-1} (V_i \times V_{i+1}).
$$
(All edges go from slice $V_i$ to slice $V_{i+1}$ for $i=0,1,2,\ldots,m-1$.)
All nodes excepted those in the last slice $V_m$ have a positive outdegree.
Nodes in $V_{\exists}$ are called \emph{existential} nodes,
and nodes in $V_{\forall}$ are called \emph{universal} nodes.
Notice that $V_0\subseteq V_{\exists}$ by definition.
Alternating paths from node $x$ to nodes in $T\subseteq V_m$ are defined as follows by
the property $\apath_G(x,T)$.

\begin{enumerate}
	\item[(1)] for $x\in V_m$ $\apath_G(x,T)$ iff $x\in T$
	\item[(2a)] for $x\in V_{\exists}-V_m : \apath_G(x,T)$ iff $\exists z\in V_{\forall} : (x,z)\in E$ and $\apath_G(z,T)$
	\item[(2b)] for $x\in V_{\forall}-V_m : \apath_G(x,T)$ iff $\forall z\in V_{\exists} :$ if $(x,z)\in E$ then $\apath_G(z,T)$
\end{enumerate}

The problem \ASAGAP is similar to the alternating graph accessibility problem,
but for the restricted class of alternating slice graphs.
%
\problemdef{\ASAGAP}{The alternating slice graph accessibility problem}{$\langle G,s,T \rangle$, where $G=(V_{\exists}\cup V_{\forall},E)$ is an alternating slice graph
                                                  with slices $V_0,V_1,\ldots,V_{m}$, $m$ even,
                                          and $s\in V_{0}$, $T\subseteq V_m$}{Does $\apath_G(s,T)$ hold}

We will use also the following variant where the outdegree of $\forall$-nodes
and the indegree of $\exists$-nodes is restricted.

\problemdef{\ASAGAPtwooutonein}{The alternating slice graph accessibility problem with bounded degree}{$\langle G,s,T \rangle$, where $G=(V_{\exists}\cup V_{\forall},E)$ is an alternating slice graph
                                                  with slices $V_0,V_1,\ldots,V_{m}$,
                                                  where every node in $V_{\forall}$ has outdegree $2$ and every node in $V_{\exists}-V_0$ has indegree $1$,
                                          and $s\in V_{0}$, $T\subseteq V_m$}{Does $\apath_G(s,T)$ hold}



$\ASAGAP_{\log}$ is the set of all elements $\langle G,s,T\rangle$ of $\ASAGAP$,
where $G$ is a graph with $n$ nodes and $m$ slices such that $m\leq \log n$.
Similarly, the problem $\ASAGAPtwooutonein_{\log}$ is the subset of $\ASAGAPtwooutonein$ with graphs of logarithmic depth.
The following completeness results are straightforward.

\begin{theorem}
\begin{enumerate}
\item  $\ASAGAP$ is $\P$-complete \cite{MW-RP10}.
\item $\ASAGAPtwooutonein$ is $\P$-complete.
\item $\ASAGAP_{\log}$ is $\AC1$-complete \cite{MW13}.
\item $\ASAGAPtwooutonein_{\log}$ is $\LOGCFL$-complete.
\end{enumerate}
\end{theorem}

A Kripke model $(W,R,\xi)$ contains a total graph $(W,R)$.
We will use several methods to transform an alternating graph
to a graph that appears as (part of) a Kripke model.

If $G=(V,E)$ is an alternating graph with slices $V_0,\ldots,V_m$,
then $G^{\sharp}=(V^{\sharp},E^{\sharp})$ is the total graph
obtained from $G$ by adding a singleton slice $V_{m+1}=\{e\}$
and edges from all nodes in $V_m\cup V_{m+1}$ to $e$.
More formally, $V^{\sharp}=V\cup V_{m+1}$ and
$E^{\sharp} = E \cup ((V_m\cup\{e\})\times \{e\})$.

For instances of $\ASAGAPtwooutonein$, we will also apply another transformation.
Let $G=(V,E)$ with slices $V_0,\ldots,V_m$ be such an instance.
Every slice $V_i\subseteq V_{\exists}-V_0$ consists of
nodes with indegree $1$.
(Remind that node(s) in $V_0$ have indegree $0$.)
Thus $V_i\subseteq V_{\exists}-V_0$ can be considered as being partitioned
into sets $V^u_i:=\{v\mid (u,v)\in E\}$ for every $u\in V_{i-1}$.
Then each $V^u_i$ consists of two nodes
which can be assumed to be ordered arbitrarily.
and we will use the notation $V^u_i = \{v_{u,1},v_{u,2}\}$.

Let $\hat{V}_i:=\{\hat{v} \mid v\in V_i\}$ be a set
of nodes that are ``copies'' of the nodes of $V_i$.
Similarly as $V_i$ for even $i>0$ (i.e. $V_i\subseteq V_{\exists}-V_0$),
$\hat{V}_i$ is partitioned into sets $\hat{V}^u_i = \{\hat{v}_{u,1},\hat{v}_{u,2}\}$
for all $u\in V_{i-1}$.
The graph $G^{\flat}=(V^{\flat},E^{\flat})$ obtained from $G$ is defined as follows.
(See also Figure~\ref{fig:Bsp-Gl} for an example.)
\vspace*{1ex}

\noindent
{
\renewcommand{\arraystretch}{1.2}
$
\begin{array}{@{}rp{0.89\textwidth}}
V^{\flat} := & $V \cup \bigcup_{i=0}^m \hat{V}_i$ \\
E^{\flat} := & $E \cap V_{\exists}\times V_{\forall}$ \hfill \small{(The edges leaving $\exists$-nodes are as in $G$.)} \\
       & $\cup ~ \{(u,v_{u,1}) \mid u\in V_{\forall}\}$ \hfill \small{($\forall$-nodes have an edge to their ``first'' successor in $G$.)}\\
       & $\cup ~ \{(v_{u,1}, \hat{v}_{u,1}),(\hat{v}_{u,1}, {v}_{u,2}),(v_{u,2}, \hat{v}_{u,2}),(\hat{v}_{u,2}, \hat{v}_{u,2}) \mid u\in V_{\forall}\}$ \\
       & \hfill \small{(From each first suc.\ $v_{u,1}$ starts a path $v_{u,1},\hat{v}_{u,1},v_{u,2},\hat{v}_{u,2}$ ending in a loop.)} \\
       & $\cup ~ \{(u,\hat{u}),(\hat{u},\hat{u}) \mid u\in V_{\forall}\cup V_0\}$ \\
       &       \hfill \small{($\forall$-nodes and $V_0$-nodes $u$ have another edge to $\hat{u}$ having a loop.)}\\
\end{array}
$
}

We will use the notion of slices also for $G^{\flat}$,
even though there are edges between nodes in the same slice.
The set of nodes $G^{\flat}$ is partitioned to $G^{\flat}=V^{\flat}_0\cup V^{\flat}_1 \cup \ldots \cup V^{\flat}_m$, where
slice $V^{\flat}_i = V_i \cup \hat{V}_i$.

\begin{figure}
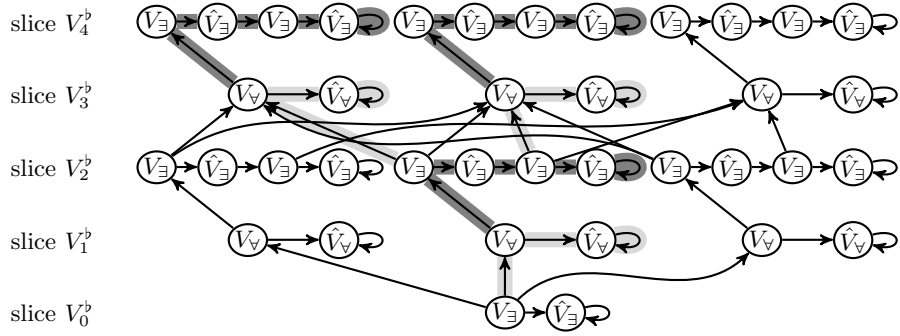


\transformedGl

\caption{The graph $G^{\flat}$ obtained from the graph $G$ in Figure~\ref{fig:Bsp-ASAGAPveeonein}.
The labels in the nodes indicate to which partition the node belongs.
The marked edges indicate infinite paths whose collection ``simulates'' the witness for $\apath_G(s,T)$ in $G$. 
}
\label{fig:Bsp-Gl}
\end{figure}

%% file: EU.tex
\subsection{Existential Until $\EU$}
\label{sec:EU}

It was shown in \cite{Beyer11} that
$\CTLMC{\EU,\land,\lor}$ is $\P$-complete.
We improve this result by showing that the Boolean operators are not necessary for the hardness and show that $\CTLMC{\EU}$ is $\P$-complete (Theorem~\ref{thm:EU}).
Since model checking for formulas with $\EU$ as single operator reaches the maximal hardness, $\EU$ turns out to be the hardest $\CTL$ operator.
We also can conclude that $\CTLMC{T}$ is $\P$-complete for every set $T$ of Boolean functions and \CTL operators that contain $\EU$.

Technically, the proof of Theorem~\ref{thm:EU} can be seen as a guide for the $\P$-hardness proofs for $\CTLMC{\ER,\lor}$ and for $\CTLMC{\EG,\oplus}$.
Since the latter consider fragments with a combination of temporal and Boolean operators, their proofs are technically more involved, but the basic strategies are similar.

\begin{theorem}\label{thm:EU}
$\CTLMC{\EU}$ is $\P$-complete.
\end{theorem}

\begin{figure}
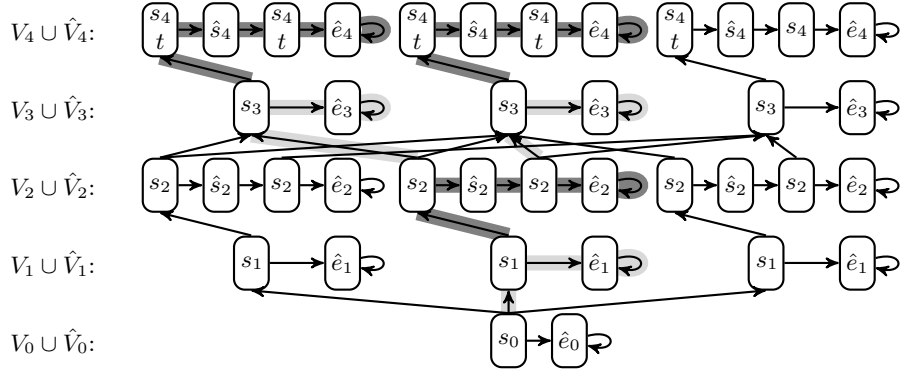

\EUfromGl
\caption{Example for the construction of $\KEU$ in the proof of Theorem~\ref{thm:EU}. The marked edges indicate the paths according to Claim~\ref{claim:EU-paths}.}
\label{fig:ex-EU}
\end{figure}

\begin{proof}
The upper bound $\P$ follows from \cite{clemsi86}.
For the lower bound---$\P$-hardness---we give a reduction from the $\P$-complete problem $\ASAGAPtwooutonein$.
Let $\langle G,s,T \rangle$ be an instance of $\ASAGAPtwooutonein$ with $G=(V,E)$ for $V=V_{\exists}\cup V_{\forall}$ with slices $V=V_0\cup\ldots\cup V_m$.
Let $G^{\flat}=(V^{\flat},E^{\flat})$ be the graph obtained from $G$ as described in Section \ref{subsubsec:Kripke}.
Using $G^{\flat}$, we construct a Kripke model $\KEU=(V^{\flat},E^{\flat},\xi)$ with assignment $\xi$ as follows (see Figure~\ref{fig:ex-EU} for an example).
\begin{enumerate}
\item $t$ is assigned to every node in $T$.
\item $s_i$ is assigned to every node in $V_i$ (for $i=0,1,\ldots,m$).
\item $\hat{s}_i$ is assigned to every node $v\in\hat{V}_i$ with $(v,v)\not\in E^{\flat}$ (for $i=0,1,\ldots,m$).
\item $\hat{e}_i$ is assigned to every node $v\in\hat{V}_i$ with $(v,v)\in E^{\flat}$ (for $i=0,1,\ldots,m$).
\end{enumerate}

The formulas $\phi_i$ are defined inductively for $i=m,m-1,\ldots,0$ as follows.
\begin{align*}
 \phi_i &:= 
\begin{cases}
 t, &\text{ if }i=m,\\
 s_i \EU ((\hat{s}_{i+1} \EU \phi_{i+1}) \EU \hat{e}_{i+1}),  & \text{ if $i<m$}.
\end{cases}
\end{align*}

The Kripke model $\KEU$ and the formulas $\phi_i$ are constructed in a way that simulates alternating graphs as follows.
Examples for the paths used in the following Claim are indicated by marked edges in Figure~\ref{fig:ex-EU}.

\begin{ownClaim}\label{claim:EU-paths}
\begin{enumerate}
\item Let $w\in V_i\cap V_{\exists}$ for some $i<m$.
Then $\KEU,w\models\phi_i$ if and only if
there exists a $\pi\in\Pi(w)$ with $\pi[2]\in V_{i+1}$
such that $\KEU,\pi[2]\models \phi_{i+1}$.
\item Let $w\in V_i\cap V_{\forall}$ for some $i<m$.
Then $\KEU,w\models\phi_i$ if and only if
there exists a $\pi\in\Pi(w)$ with $\pi[2],\pi[4]\in V_{i+1}$
such that $\KEU,\pi[2]\models \phi_{i+1}$ and $\KEU,\pi[4]\models \phi_{i+1}$.
\end{enumerate}
\end{ownClaim}

Now we only have to use the relation between $G$ and $G^{\flat}$.

\begin{ownClaim}\label{claim:EU-main}
For every $i\leq m$ and every $w\in V_i$ holds: $\KEU,w\models \phi_i$ if and only if $\apath_G(w,T)$.
\end{ownClaim}

The proofs of the above claims can be found in the Appendix.
With Claim~\ref{claim:EU-main} we get that $\langle G,s,T\rangle \in \ASAGAPtwooutonein$ if and only if $\KEU,s\models\phi_0$.
The $\CTLMC{\EU}$ instance $\langle \KEU,s, \phi_{0}\rangle$ can be computed in space logarithmic in the size of $G$.
Thus we have shown $\ASAGAPtwooutonein\leqlogm\CTLMC{\EU}$.\qed
\end{proof}

From Theorems \ref{thm:CTLMC} and \ref{thm:EU} we immediately get the 
complete characterization of the complexity of model checking for fragments with $\EU$---i.e., the model checking fingerprint of $\EU$.

\begin{theorem}\label{EU-fingerprint}
$\CTLMC{\EU,B}$ is $\P$-complete for every $B\subseteq\{\neg,\land,\lor,\oplus\}$.
\end{theorem}

%% file: remainingOperators.tex
\section{The Remaining Existential Operators: Release $\ER$, Globally $\EG$, Next $\EX$, and Future $\EF$}\label{sec:remaining}
Let us first turn to the case of existential next $\EX$ as nothing has to be proven. Its model checking fingerprint actually is already known even though it is not always stated in the way we do it here.

\begin{theorem}\label{thm:ctlmc-EX}\label{EX-fingerprint}
Let $B\subseteq\{\neg,\land,\lor,\oplus\}$. Then $\CTLMC{\EX,B}$ is\vspace{-1.5ex}
\begin{enumerate}
\item $\P$-complete for $B\supseteq\{\neg\}$ or $B\supseteq\{\oplus\}$ \cite{sch02},
\item $\LOGCFL$-complete for $B=\{\land,\lor\}$ or $B=\{\land\}$ \cite{Beyer11}, and 
\item $\NL$-complete for $B\subseteq\{\lor\}$ (follows immediately from \cite[Theorem 3.3]{mc_survey_beatcs}).
\end{enumerate}
\end{theorem}

For the remainder of the results in this section we have to omit the proofs due to space constraints. However the details are all presented in the appendix. This section is structured as follows. We will state a model checking fingerprint theorem and then start to explain and discuss the results in order to give some intuition on the proof technique. Also we will mention some connections between the results, e.g., how the overall picture presents. Let us begin with the fingerprint of $\ER$.

\begin{theorem}\label{ER-fingerprint}
Let $B\subseteq\{\neg,\land,\lor,\oplus\}$.
Then $\CTLMC{\ER,B}$ is\vspace{-1.5ex}
\begin{enumerate}
\item $\P$-complete for $B\supseteq\{\lor\}$ or $B\supseteq\{\neg\}$ or $B\supseteq\{\oplus\}$,
\item $\LOGCFL$-hard for $B\supseteq\{\land\}$, and
\item $\LOGCFL$-complete for $B=\emptyset$.\label{itm:ER-empty}
\end{enumerate}
\end{theorem}

The $\P$-completeness of $\CTLMC{\ER,\land,\lor}$ is shown in \cite{Beyer11}.
We improve this result by showing $\P$-hardness already for $\CTLMC{\ER,\lor}$. The optimality of this hardness result is witnessed by the $\LOGCFL$-completeness of $\CTLMC{\ER}$. Also observe that this shows that $\ER$ is not as powerful as $\EU$. Our results are completed by the $\P$-hardness of $\CTLMC{\ER,\lnot}$. Concluding, this shows that $\ER$ is strictly simpler than $\EU$ (unless $\LOGCFL=\P$). 

Let us now consider the case of $\ER$ with $\neg$. If $s$ is an atom that is satisfied only by a node $w$ and all its successors in a Kripke model $K$, and no successor of $w$ satisfies $s$, then $K,w \models \alpha \ER s$ if and only if $K,w \models \EX \alpha$. One can use this idea to translate the $\P$-hardness proof of $\CTLMC{\EX,\neg}$ to a $\P$-hardness proof of $\CTLMC{\ER,\neg}$. 

We settled complete characterizations of the complexity of $\CTLMC{\ER,B}$ for fragments with $\ER$ as only \CTL operator for all $B\subseteq\{\neg,\land,\lor,\oplus\}$ except $B=\{\land\}$.
For $\CTLMC{\ER,\land}$ we have $\LOGCFL$-hardness and containment in $\P$ (follows from \cite{Beyer11}).
A result with matching upper and lower bounds yet remains open.\medskip

Turning to the case of existentially globally operator $\EG$, interestingly, this operator combines an existential and universal quantification in a single operator. This is worth noting as it proves  itself as powerful operator from complexity point of view. 

\begin{theorem}\label{EG-fingerprint}
Let $B\subseteq\{\neg,\land,\lor,\oplus\}$.
Then $\CTLMC{\EG,B}$ is\vspace{-1.5ex}
\begin{enumerate}
\item $\P$-complete for $B\supseteq\{\land,\lor\}$ or $B\supseteq\{\oplus\}$, and
\item $\NL$-complete for $B\subseteq\{\land\}$ or $B\subseteq\{\lor\}$ or $B\subseteq\{\neg\}$.
\end{enumerate}
\end{theorem}

It was shown in \cite{Beyer11} that $\CTLMC{\EG,\land,\lor}$ is $\P$-complete. We prove that this result is optimal by showing that $\CTLMC{\EG,\land}$ and $\CTLMC{\EG,\lor}$ are both $\NL$-complete. Further we obtain the same characterization for $\CTLMC{\EG}$ and $\CTLMC{\EG,\lnot}$. The most intriguing result is the $\P$-completeness of the fragment $\CTLMC{\EG,\oplus}$. One can reduce from $\ASAGAPtwooutonein$ but it is quite demanding to explicitly argue on the chosen paths depending on the occurring exclusive-ors $\xor$. From the model construction one can easily see similarities to the one shown in Figure~\ref{fig:ex-EU} however one needs additional propositions labelled on the states to ensure having control on the paths. 

\begin{theorem}\label{EF-fingerprint}
Let $B\subseteq\{\neg,\land,\lor,\oplus\}$.
Then $\CTLMC{\EF,B}$ is\vspace{-1.5ex}
\begin{enumerate}
\item $\P$-complete for $B\supseteq\{\land,\oplus\}$, 
\item $\AC1$-hard for $B\supseteq\{\oplus\}$, and
\item $\LOGCFL$-complete for $\{\land\}\subseteq B\subseteq\{\land,\lor\}$, and
\item $\NL$-complete for $B\subseteq\{\lor\}$ or $B\subseteq\{\neg\}$.
\end{enumerate}
\end{theorem}

In \cite{Beyer11} it is shown that $\CTLMC{\EF,\wedge,\vee}$ is $\LOGCFL$-complete. Since their hardness proof does not use $\vee$, it follows that $\CTLMC{\EF,\wedge}$ is $\LOGCFL$-complete, too.
Moreover, $\CTLMC{\EF,\AG,\wedge,\vee}$ is shown to be $\P$-complete in \cite{Beyer11}; we get $\P$-completeness for $\CTLMC{\EF,\wedge,\oplus}$. We classified almost all remaining cases.
We show that the cases $\CTLMC{\EF}$, $\CTLMC{\EF,\neg}$, and $\CTLMC{\EF,\vee}$ are all $\NL$-complete.
However the most interesting result is the $\AC1$-hardness of $\CTLMC{\EF,\oplus}$. There are only very few problems known for which $\AC1$ is the best shown lower bound. In fact, Cook~\cite{Cook85} asks for natural $\AC1$-complete problems, i.e., problems where the $\AC1$-completeness is not forced by some logarithmic bounds in the problem definition. 
Chandra and Tompa~\cite{ChandraT90} show an $\AC1$-complete two-person-game that has $\AC1$ as a straightforward upper bound and continue to ask for ``less straightforward'' $\AC1$-complete problems.
One such problem is the model checking problem for intuitionistic logic with one atom \cite{MW13}.
The model checking problem for the $\{\EF,\oplus\}$-fragment is a very hot candidate.
Anyway, it seems to be a very challenging question to show whether this problem belongs to Cook's list.

%% file: conclusion.tex
\section{Conclusion}
\label{sec:conclusion}

In this paper we aimed to present a complete complexity classification of all 
fragments of \CTL with one \CTL operator and arbitrary Boolean functions.
An overview of the complexity results is given in Figure~\ref{fig:complexity-overview}.
We stated all our results for \CTL operators that start with the existential path quantifier $\E$.
But our classification easily generalizes to the remaining \CTL operators
starting with the universal path quantifier $\A$ through the well-known dualities.
%
%
%
Simply said, if $\CTLMC{T,B}$ is complete (resp.~hard) for a complexity class $\mathcal C$,
then $\CTLMC{\dual{T},\dual{B}}$ is complete (resp.~hard) for the complement $\text{co-}\mathcal C$ of $\mathcal C$.
Thus, e.g., from the $\AC1$-hardness of $\CTLMC{\EF,\oplus}$ (Theorem~\ref{EF-fingerprint}) we immediately obtain $\AC1$-hardness of $\CTLMC{\AG,\oplus}$, and from $\LOGCFL$-complete\-ness of $\CTLMC{\EX,\land}$ (Theorem~\ref{thm:ctlmc-EX}) we obtain $\LOGCFL$-completeness of the corresponding $\CTLMC{\AX,\lor}$.
Our results can directly be rewritten to deal not only with Boolean operators but in a more generalized view with Boolean clones as, e.g., in the work of Bauland et~al.\ and Beyersdorff et~al.~\cite{BaulandLTL11,Beyer11}.

The only open cases for which we yet cannot  prove matching upper and lower bounds are $\CTLMC{\ER,\land}$ and $\CTLMC{\EF,\xor}$ and, of course, their duals $\CTLMC{\AU,\lor}$ and $\CTLMC{\AG,\xor}$.
Although we could not achieve an $\AC1$ upper bound for $\CTLMC{\EF,\xor}$, we are convinced that such a result seems closer than proving $\P$-hardness (or some stronger hardness result than $\AC1$). 

Our classifications can be applied to compare the expressiveness of single \CTL operators with respect to the complexity of the induced model checking problems.

\begin{definition}
Let $S$ and $T$ be a set of \CTL operators. We say that $T$ is \emph{mc-stronger} than $S$ (abbreviated as $S\morepowerfull T$), if for all sets $B$ of Boolean functions holds $\CTLMC{S,B} \leqlogm \CTLMC{T,B}$.
\end{definition}

The reflexive and transitive relation $\morepowerfull$ for mc-strength compares what we informally called the model checking fingerprints of \CTL operators. 
Our fingerprint theorems (Theorems~\ref{EU-fingerprint}--\ref{EF-fingerprint}) yield the hierarchy of mc-strength of \CTL operators shown in Figure~\ref{fig:OperatorHierarchie}.
The notion of mc-strength generalizes the notion of expressiveness \cite{laroussinie95} of \CTL operators. For example, since $\EG\alpha\equiv 0 \ER \alpha$,
$\ER$ is more expressive than $\EG$. With our notion we obtain also $\EG\morepowerfull \ER$.
But our notion yields more information about differences between several operators. For example, $\EX$ and $\EU$ have incomparable expressiveness, but we obtain $\EX \morepowerfull \EU$.

A strength-relation like $\morepowerfull$ can also be defined with respect to the satisfiability problem---call it \emph{sat-strength}.
Whereas for model checking the set of \CTL operators is partitioned into seven sets with different mc-strength (see Figure~\ref{fig:OperatorHierarchie}), from~\cite{mmtv09,meier11} it follows that the comparison by sat-strength yields only the following three partitions with increasing strength:
$\{\AF,\EG\}$, $\{\EX,\AX,\EF,\AG\}$, and $\{\AU,\EU,\ER,\AR\}$.
The three notions expressiveness, sat-strength, and mc-strength intuitively compare as follows.
Expressiveness relies on \emph{equivalence} of formulas, sat-strength relies on \emph{equisatisfiability} of formulas, and mc-strength on \emph{equisatisfaction} of model checking instances.

Further work should solve the exact complexity of $\CTLMC{\EF,\oplus}$, which seems to be a very challenging problem.
Moreover, one should study the mc-strength of other temporal logics or of pairs of \CTL operators.

%% file: Appendix-EU.tex
\subsection{$\EU$}

\textbf{Proofs for Theorem \ref{thm:EU}:}
\emph{$\CTLMC{\EU}$ is $\P$-complete.}
\vspace*{2ex}

The basic semantical property of $\EU$ that we will use is
\begin{multline}\label{EUprop}
K,w \models \alpha \EU \beta \text{ if and only if } \\
  (i)~K,w \models \beta \text{ or } 
  (ii)~ K,w \models \alpha \text{ and } K,v\models \alpha\EU \beta \text{ for a successor $v$ of $w$.} 
\end{multline}

\begin{ownClaim}\label{claim:eu-upper}
For all $i<m$, all nodes $w\in V^{\flat}_i$, and all $j>i$  holds: $\KEU,w \nmodels \phi_{j}$.
\end{ownClaim}

Proof. Let $i<m$ and $w\in V^{\flat}_i$.
We proceed by induction on $j=m,m-1,\ldots,i+1$.
The base case is clear since $t\not\in\xi(w)$ and thus $\KEU,w \nmodels t (=\phi_{m})$.
For $j<m$, we have $\KEU,w\nmodels \phi_{j+1}$ as inductive hypothesis.
Assume $\KEU,w\models s_j \EU ((\hat{s}_{j+1} \EU \phi_{j+1}) \EU \hat{e}_{j+1}) ~(=\phi_j)$.
The next steps use (\ref{EUprop}).
From $s_j\not\in\xi(w)$ follows $\KEU,w\models (\hat{s}_{j+1} \EU \phi_{j+1}) \EU \hat{e}_{j+1}$.
From $\hat{e}_{j+1}\not\in\xi(w)$ then follows $\KEU,w\models \hat{s}_{j+1} \EU \phi_{j+1}$,
and from $\hat{s}_{j+1}\not\in\xi(w)$ we conclude $\KEU,w\models \phi_{j+1}$.
This contradicts the inductive hypothesis.
Thus $\KEU,w\nmodels\phi_j$.
\claimqed

The Kripke model $\KEU$ and the formulas $\phi_i$ are constructed
in a way that simulates alternating graphs as follows.
\vspace*{2ex}

\noindent
\textbf{Claim \ref{claim:EU-paths}:}
\vspace{-4ex}

\emph{
\begin{enumerate}
\item Let $w\in V_i\cap V_{\exists}$ for some $i<m$.
Then $\KEU,w\models\phi_i$ if and only if
there exists a $\pi\in\Pi(w)$ with $\pi[2]\in V_{i+1}$
such that $\KEU,\pi[2]\models \phi_{i+1}$.
\item Let $w\in V_i\cap V_{\forall}$ for some $i<m$.
Then $\KEU,w\models\phi_i$ if and only if
there exists a $\pi\in\Pi(w)$ with $\pi[2],\pi[4]\in V_{i+1}$
such that $\KEU,\pi[2]\models \phi_{i+1}$ and $\KEU,\pi[4]\models \phi_{i+1}$.
\end{enumerate}
}
\vspace*{2ex}

Proof of (1).
For the proof direction from left to right, assume $\KEU,w\models\phi_i$.
Since $s_i\in\xi(w)$ and $s_i\not\in\xi(v)$ for all successors $v$ of $w$, 
it follows from (\ref{EUprop}) that $\KEU,v\models (\hat{s}_{i+1} \EU \phi_{i+1}) \EU \hat{e}_{i+1}$
for some successor $v$ of $w$.
For $v'\in V^{\flat}_i$ holds $\hat{s}_{i+1}, \hat{e}_{i+1}\not\in \xi(v)$ and $\KEU,v'\nmodels \phi_{i+1}$ (Claim \ref{claim:eu-upper}),
and using (\ref{EUprop}) we get $\KEU,v'\nmodels (\hat{s}_{i+1} \EU \phi_{i+1}) \EU \hat{e}_{i+1}$.
Thus there is a successor $v\in V_{i+1}$ of $w$ with $\KEU,v\models(\hat{s}_{i+1} \EU \phi_{i+1}) \EU \hat{e}_{i+1}$.
Since $\hat{s}_{i+1}, \hat{e}_{i+1}\not\in\xi(v)$, this means $\KEU,v\models\phi_{i+1}$.

For the other proof direction assume $\KEU,\pi[1]\models s_i$ and $\KEU,\pi[2]\models \phi_{i+1}$
for some $\pi\in\Pi(w)$ with $\pi[2]\in V_{i+1}$.
Using (\ref{EUprop}) it follows that $\KEU,\pi[2]\models \hat{s}_{i+1} \EU \phi_{i+1}$.
Moreover, $\pi[2]$ has a successor $u\in \hat{V}_{i+1}$ with $\KEU,u\models \hat{e}_{i+1}$.
Thus from (\ref{EUprop}) follows $\KEU,\pi[2]\models (\hat{s}_{i+1} \EU \phi_{i+1}) \EU \hat{e}_{i+1}$.
Since $\pi[2]$ is a successor of $\pi[1](=w)$ and  $s_i\in\xi(w)$, we get $\KEU,w\models s_i \EU ((\hat{s}_{i+1} \EU \phi_{i+1}) \EU \hat{e}_{i+1})$.

Proof of (2).
This can be shown using similar arguments as above.
\claimqed
\vspace*{2ex}

\noindent
\textbf{Claim \ref{claim:EU-main}:}
\emph{For every $i\leq m$ and every $w\in V_i$ holds:
$\KEU,w\models \phi_i$ if and only if $\apath_G(w,T)$.
}
\vspace*{1ex}

Proof.
The proof proceeds by induction on $i$.
The base case for nodes in slice $i=m$ is straightforward. 

For the inductive step, we consider $i<m$ and $w\in V_i$.

First, consider $w\in V_{\exists}$.
By Claim \ref{claim:EU-paths}, we have that 
$\KEU,w\models \phi_i$ if and only if $w$ has a successor $v\in V_{i+1}$ with $\KEU,v\models \phi_{i+1}$.
Since $v$ is also a successor of $w$ in $G$,
using the inductive hypothesis, the latter is equivalent to $\apath_G(v,T)$ for a successor $v$ of $w$ in $G$.
Since $w\in V_{\exists}$, this means $\apath_G(w,T)$.

Next, consider $w\in V_{\forall}$.
By Claim \ref{claim:EU-paths}, we have that 
$\KEU,w\models \phi_i$ if and only if there exists a path $\pi\in\Pi(w)$ with $\pi[2],\pi[4]\in V_{i+1}$
such that $\KEU,\pi[2]\models \phi_{i+1}$ and $\KEU,\pi[4]\models\phi_{i+1}$.
Since $\pi[2]$ and $\pi[4]$ are all successors of $w$ in $G$,
using the inductive hypothesis, the latter is equivalent to
$\apath_G(v,T)$ for all successors $v$ of $w$ in $G$.
Since $w\in V_{\forall}$, this means $\apath_G(w,T)$.
\claimqed

%% file: Appendix-ER.tex
\subsection{$\ER$}
\begin{figure}
\ERveefromGl
\caption{Kripke model $\KER$ obtained from the $\ASAGAPtwooutonein$ instance in Figure~\ref{fig:Bsp-ASAGAPveeonein}.}
\label{fig:KER}
\end{figure}

\begin{theorem}
$\CTLMC{\ER,\lor}$ is $\P$-complete. 
\end{theorem}

\begin{proof}
Containment in $\P$ follows from \cite{clemsi86}.
In order to show $\P$-hardness, we logspace reduce from $\ASAGAPtwooutonein$.
Let $\langle G,s,T\rangle$ be an instance of $\ASAGAPtwooutonein$, where $G=(V,E)$ and $V=V_\exists\cup V_{\forall}$
consists of slices $V_0,V_1,\ldots,V_m$ with $s\in V_0$ and $T\subseteq V_m$.
Let $G^{\flat}=(V^{\flat},E^{\flat})$ be the graph obtained from $G$ as described in Section~\ref{subsubsec:Kripke}.
In order to obtain the Kripke model $\KER=(V^{\flat},E^{\flat},\xi)$, it remains to define the assignment $\xi$ of atoms to sets of nodes of $G^{\flat}$.
\begin{enumerate}
\item $t$ is assigned to every node in $T$.
\item $s_0$ is assigned to every node in $V_0$.
\item $s_{i}$ and $s_{i-1}$ are assigned to every node in $V_i$ for $i>0$.
       
\item $t_{i-1}$ is assigned to every node $v\in V_{\exists}\cap V_i$ ($i>0$) that is the second successor $v=v_{u,2}$
      of a node $u\in V_{\forall}\cap V_{i-1}$.
       
\item $\hat{s}_{i-1}$ is assigned to every $v\in \hat{V}_i$ ($i>0$) such that $v=\hat{v}_{u,1}$
      is the copy of the first successor of a node $u\in V_{\forall}\cap V_{i-1}$.
\end{enumerate}
Nothing is assigned to nodes $\hat{v}_{u,2}$ and to nodes in $\hat{V}_{\forall}$.
Notice that all infinite paths in $G_{\flat}$ must eventually loop in a state $\hat{v}_{u,2}$ that satisfies no atom at all.
See Figure~\ref{fig:KER} for an example.

The formulas $\phi_i$ are defined inductively as follows, for $i=m,m-1,\ldots,0$.
\begin{align*}
 \phi_i &:= 
\begin{cases}
 t, &\text{ if }i=m,\\
 \phi_{i+1}\ER s_{i},  & \text{ if $i<m$ and even (slice with $\exists$-nodes)},\\
 t_{i} \ER ((\phi_{i+1}\ER s_{i}) \lor \hat{s}_i),  & \text{ if $i<m$ and odd (slice with $\forall$-nodes)}.
\end{cases}
\end{align*}

The following claim states that $\apath_G(x,T)$
corresponds to the satisfaction of formulas in the corresponding nodes in the constructed Kripke model.

\begin{ownClaim}\label{claim:ER-or-main}
For every $l\leq m$ and every $v\in V_l$ holds: 
$\apath_G(v,T)$ if and only if $\KER,v\models \phi_l$.
\end{ownClaim}

\begin{proof}
 The proof proceeds by induction on $l$.
The base case for  slice $l=m$ is straightforward.
For the inductive step, consider $l<m$.
First consider even $l$ (slice of $\exists$-nodes),
take a node $v\in V_l\cap V_{\exists}$,
and assume 
\begin{align}
\KER,v \models \phi_l, ~~~~\text{i.e.~} \KER,v \models \phi_{l+1}\ER s_{l} . \label{ER0}
\end{align}
Remind that $\phi_{l+1}=t_{l+1} \ER ((\phi_{l+2}\ER s_{l+1}) \lor \hat{s}_{l+1})$.
Since $s_{l+1}, \hat{s}_{l+1}\not\in\xi(v)$,
it follows that $\KER,v \nmodels \phi_{l+1}$.
Since $s_{l}\in\xi(v)$, it follows that (\ref{ER0}) is equivalent to
\begin{align}
\text{there exists a successor $z$ of $v$ with ~} \KER,z & \models \phi_{l+1} \ER s_l . \label{ER2}
\end{align}
The successor $\hat{v}\in\hat{V}_l$ of $v$ does not satisfy $s_l$.
Thus (\ref{ER2}) is equivalent to
\begin{align}
\text{there exists a successor $z\in V_{l+1}$ of $v$ with ~} \KER,z & \models \phi_{l+1} \ER s_l . \label{ER3}
\end{align}
All successors $z\in V_{l+1}$ of $v$ satisfy $s_l$
and do not have a successor that satisfies $s_l$.
Thus (\ref{ER3}) is equivalent to
\begin{align}
\text{there exists a successor $z\in V_{l+1}$ of $v$ with ~} \KER,z & \models \phi_{l+1} . \label{ER4}
\end{align}
By the inductive hypothesis and the construction of the Kripke model,
this means that $v$ has a successor $z$ in $G$ with $\apath_G(z,T)$.
Since $v\in V_{\exists}$, this is equivalent to $\apath_G(v,T)$.
 
Now consider odd $l$  (slice of $\forall$-nodes),
take a node $u\in V_l\cap V_{\forall}$,
and assume 
\begin{align}
\KER,u \models \phi_l, ~~~\text{i.e.~} \KER,u \models t_{l} \ER ((\phi_{l+1}\ER s_{l}) \vee \hat{s}_l) . \label{ER4-5}
\end{align}
There is no infinite path in $\Pi(u)$ that satisfies $s_l$ or $\hat{s}_l$ in every node.
Thus, 
(\ref{ER4-5}) can only be witnessed by an infinite path that passes through a node that satisfies $t_{l}$.
Every such path in $\Pi(u)$ has the finite prefix
$u,v_{u,1}, \hat{v}_{u,1},v_{u,2}$.
Thus (\ref{ER4-5}) is equivalent to
\begin{align}
\text{every node in the path  $u,v_{u,1}, \hat{v}_{u,1},v_{u,2}$
satisfies $(\phi_{l+1}\ER s_{l}) \lor \hat{s}_l$.} \label{ER5}
\end{align}
For $\hat{v}_{u,1}$ this holds, since $\hat{s}_l\in\xi(\hat{v}_{u,1})$.

$\KER,u\models s_l$, but $\KER,u\nmodels \phi_{l+1}$ because $\KER,u\nmodels s_{l+1}$.
Thus $\KER,u\models (\phi_{l+1}\ER s_{l}) \lor \hat{s}_l$ if and only if 
   $\KER,v_{u,1} \models \phi_{l+1}\ER s_l$ and $\KER,v_{u,2} \models \phi_{l+1}\ER s_l$.
Since $\KER,v_{u,i}\models s_l$ and no successor $\hat{v}_{u,i}$ of $v_{u,i}$ satisfies $s_l$,
it follows that $\KER,v_{u,i} \models \phi_{l+1}\ER s_l$ is equivalent to having $\KER,v_{u,i}\models \phi_{l+1}$ (for $i=1,2$).
This yields that (\ref{ER5}) is equivalent to
\begin{align}
\text{$\KER,v_{u,1} \models \phi_{l+1}$ and $\KER,v_{u,2} \models \phi_{l+1}$.} \label{ER6}
\end{align}
By the inductive hypothesis and the construction of the Kripke model, 
this means that $\apath_G(v,T)$ holds for all successors $v$ of $u$ in $G$.
The latter is equivalent to $\apath_G(u,T)$. 
\claimqed
\end{proof}

With Claim~\ref{claim:ER-or-main}
we get that $\langle G,s,T\rangle \in \ASAGAPtwooutonein$ if and only if $\KER,s\models\phi_0$.
The $\CTLMC{\ER,\lor}$ instance $\langle \KER,s, \phi_{0}\rangle$
can be computed in space logarithmic in the size of $G$.
Thus $\ASAGAPtwooutonein$ logspace reduces to $\CTLMC{\ER,\lor}$.\qed
\end{proof}

\vspace*{3ex}

\begin{theorem}
$\CTLMC{\ER,\neg}$ is $\P$-complete. 
\end{theorem}

\begin{proof}
Containment in $\P$ follows from \cite{clemsi86}.
In order to show $\P$-hardness,
we give a reduction from $\ASAGAP$.
Let $\langle G,s,T \rangle$ be an instance of $\ASAGAP$,
where $G=(V,E)$ for $V=V_{\exists}\cup V_{\forall}$
with slices $V_0,V_1,\ldots,V_m$.
Let $G^{\sharp}=(V^{\sharp},E^{\sharp})$ be the graph obtained from $G$ 
as described in Section~\ref{subsubsec:Kripke}.
In order to define the Kripke model $\KERn=(V^{\sharp},E^{\sharp},\xi)$,
we must give a definition of the assignment function $\xi$.
\begin{itemize}
\item $t$ is assigned to all nodes in $T$.
\item $s_i$ and $s_{i-1}$ are assigned to all nodes in $V_i$ for $i=0,1,\ldots,m$.
\end{itemize}

The formulas $\phi_i$ are defined inductively for $i=m,m-1,\ldots,0$ as follows.

\begin{align*}
 \phi_i &:= 
\begin{cases}
 t, & \text{ if }i=m,\\
 \phi_{i+1}\ER s_{i},  & \text{ if $i<m$ is even (slice of $\exists$-nodes)}, \\
 \neg(\neg \phi_{i+1} \ER s_{i}),  & \text{ if $i<m$ is odd (slice of $\forall$-nodes)}.
\end{cases}
\end{align*}

Notice that in the new node $e$, no atom is satisfied,
and therefore no $\phi_i$ is satified.

\begin{ownClaim}\label{claim:ERn-main}
For all $i\leq m$ and all $v\in V_i$ holds: 
$\KERn,v \models \phi_{i}$ if and only if $\apath_G(v,T)$.
\end{ownClaim}

The induction base $i=m$ is straightforward.
For the induction step we consider $i<m$ and $v\in V_i$.
We first consider even $i$.
Since $\KERn,v\nmodels s_{i+1}$ and on all paths $\pi\in\Pi(v)$,
$s_i$ is satisfied only in $\pi[1]=v$ and $\pi[2]$,
it follows that $\KERn,v\models \phi_{i+1}\ER s_{i}$ is equivalent to 
$\KERn,w\models \phi_{i+1}$ for some successor $w$ of $v$.
By the inductive hypothesis we obtain this to be equivalent to $\apath_G(v,T)$.

Next we consider odd $i$.
Assume $\KERn,v \models \neg(\neg \phi_{i+1} \ER s_{i})$,
i.e. $\KERn,v \nmodels \neg \phi_{i+1} \ER s_{i}$.
Since $\KERn,v \models \neg \phi_{i+1}$ 
and for all $\pi\in\Pi(v)$ holds that $s_i$ is satisfied only in $\pi[1]=v$ and $\pi[2]$,
it follows that $\KERn,v \nmodels \neg \phi_{i+1} \ER s_{i}$
is equivalent to $\KERn,w \nmodels \neg \phi_{i+1}$ for all successors $w$ of $v$.
The latter means that $\KERn,w \models \phi_{i+1}$ for all successors $w$ of $v$.
Using the inductive hypothesis,
we obtain $\apath_G(v,T)$.
\claimqed

The mapping from $\ASAGAP$-instances $\langle G,s,T \rangle$ to
$\CTLMC{\ER,\lnot}$-instances $\langle \KERn,s,\phi_0\rangle$ can be computed in logarithmic space.
With Claim \ref{claim:ERn-main} this yields that $\ASAGAP$ logspace reduces to $\CTLMC{\ER,\lnot}$.
\end{proof}


\begin{figure}

\centering

\ERistLOGCFLhart

\caption{Kripke model $\KMER$ obtained from the $\ASAGAP$ instance in Figure~\ref{fig:Bsp-ASAGAPveeonein}.}
\label{fig:ex-ER}
\end{figure}

\begin{theorem}\label{thm:ctlmc-ER-LOGCFL-hard}
$\CTLMC{\ER}$ is $\LOGCFL$-hard.
\end{theorem}

\begin{proof}
We $\leqlogm$-reduce from the $\LOGCFL$-complete $\ASAGAPtwooutonein_{\log}$.
Let $\langle G,s,T \rangle$ be an instance of $\ASAGAPtwooutonein_{\log}$,
where $G=(V,E)$ with $V=V_{\exists}\cup V_{\forall}$
and slices $V_0,\ldots,V_m$ for $m\leq \log|V|$.
Let $G^{\sharp}=(V^{\sharp},E^{\sharp})$ be the graph obtained from $G$ 
as described in Section~\ref{subsubsec:Kripke}.
In order to define the Kripke model $\KMER=(V^{\sharp},E^{\sharp},\xi)$ (see Figure~\ref{fig:ex-ER} for an example),
we need to specify the assignment function $\xi$.
\begin{enumerate}
\item $t$ is assigned to all nodes in $T$.
\item $s_i$ is assigned to every node in $V_{\exists}\cap V_i$.
\item $s_{i-1}$, $s^l_i$, and $s^r_i$ are assigned to every node in $V_{\forall}\cap V_i$.
\item For $u\in V_{\forall}\cap V_i$,
      the two successors $v_l$ and $v_r$ of $u$ have $u$ as only predecessor.
      Then $s^l_i$ is assigned to $v_l$ and $s^r_i$ is assigned to $v_r$.
      
      Notice that $V_{\exists}\cap V_i$ is partitioned into two sets:
      one to which $s_i^l$ is assigned and the other to which $s_i^r$ is assigned.
\end{enumerate}

The formulas $\phi_i$ are defined inductively for $i=m,m-1,\ldots,0$ as follows.

\begin{align*}
 \phi_i &:= 
\begin{cases}
 t, &\text{ if }i=m,\\
 \phi_{i+1} \ER s_{i},  & \text{ if $i<m$ is even (slice of $\exists$-nodes)}, \\
 (\phi_{i+1} \ER s^r_{i}) \ER (\phi_{i+1} \ER s^l_{i}),  & \text{ if $i<m$ is odd (slice of $\forall$-nodes)}.
\end{cases}
\end{align*}

\begin{ownClaim}\label{claim:ER-main}
For every $i\leq m$ and every $v\in V_i$ holds: 
$\KMER,v \models \phi_{i}$ if and only if $\apath_G(v,T)$.
\end{ownClaim}

The proof of the claim can be found in the Appendix.
With Claim~\ref{claim:ER-main}
we get that $\langle G,s,T \rangle \in \ASAGAPtwooutonein_{\log}$ if and only if 
  $\langle \KMER,s,\phi_0\rangle\in\CTLMC{\ER}$.
Since the transformation can be computed in logarithmic space,
it follows that $\ASAGAPtwooutonein_{\log}$ logspace reduces to $\CTLMC{\ER}$.
\end{proof}

\begin{theorem}
$\CTLMC{\ER}$ is $\LOGCFL$-complete. 
\end{theorem}

\begin{proof}
From Theorem~\ref{thm:ctlmc-ER-LOGCFL-hard} we have $\LOGCFL$-hardness, hence only membership must be shown.

A \emph{right form} of an $\{\ER\}$-formula $\psi$ is
a sequence $\langle \alpha_1,\ldots,\alpha_m,\beta \rangle$ of $\{\ER\}$-formulas such that
$\psi = \alpha_1 \ER ( \alpha_2 \ER (\alpha_3 \ER(\cdots(\alpha_m\ER\beta))\cdots)))$.
For example, $$\psi = ( a \ER b ) \ER ( (c \ER d) \ER ( e \ER f))$$
has, amongst others, the forms 
\begin{itemize}
\item $\langle a \ER b, (c \ER d) \ER ( e \ER f)\rangle$
\item $\langle a \ER b, c \ER d, e \ER f\rangle$
\item $\langle a \ER b, c \ER d, e, f\rangle$.
\end{itemize}
The third right form with $\beta=f$ is called \emph{atomic right form},
because $f$ is an atom.

\begin{ownClaim}\label{claim:er-infinite-case}
Let $\pi$ be a path through a Kripke model $K$.
The following statements are equivalent.
\begin{enumerate}
\item $\forall i\geq 1: K,\pi[i] \models \beta$
\item $\forall i\geq 1: K,\pi[i] \models \langle \alpha_1, \ldots ,\alpha_m,\beta \rangle$
\end{enumerate}
\end{ownClaim}

The proof of the Claim proceeds by induction on $m$.
The base case $m=0$ is clear, because $\beta=\langle\beta\rangle$.
For the inductive step $m>0$, the following holds.
\begin{align*}
& \forall i\geq 1: K,\pi[i] \models \langle \alpha_1,\alpha_2,\ldots,\alpha_m, \beta\rangle \\
& \Leftrightarrow \forall i\geq 1: K,\pi[i] \models \langle \alpha_2,\ldots,\alpha_m, \beta\rangle & \text{~~~(semantics of $\ER$)} \\
& \Leftrightarrow \forall i\geq 1: K,\pi[i] \models \beta & \text{~~(by the inductive hypothesis)}
\end{align*}
\claimqed

\begin{ownClaim}\label{claim:er-finite-case}
Let $\pi$ be a path through a Kripke model $K$, and let $k$ be an integer.
The following statements are equivalent.
\begin{enumerate}
\item $\forall i\leq k: K, \pi[i] \models \langle \alpha_1,\ldots,\alpha_m,\beta\rangle$
\item $K,\pi[k] \models \langle \alpha_1,\ldots,\alpha_m,\beta\rangle$ and 
        $\forall i\leq k: K,\pi[i] \models \beta$.
\end{enumerate}
\end{ownClaim}

The proof of the Claim proceeds by induction on $m$.
The base case $m=0$ is clear, because $\beta=\langle\beta\rangle$.

For the inductive step $m>0$, we consider both proof directions separately.

``$\Rightarrow$'': 
\begin{align*}
& \forall i\leq k: K,\pi[i] \models \langle \alpha_1,\ldots,\alpha_m,\beta\rangle \\
\Rightarrow &  \forall i\leq k ~~\exists \pi'\in\Pi(\pi[i]) \\
   & ~~~~~(1) \forall j\geq 1 : K,\pi'[j] \models \langle \alpha_2,\ldots,\alpha_m,\beta\rangle \text{~~or} \\
   & ~~~~~(2) \exists l\geq 1: K,\pi'[l] \models \alpha_1 ~\&~ \forall q\leq l: K,\pi'[q] \models  \langle \alpha_2,\ldots,\alpha_m,\beta\rangle \\
   & \text{~~(semantics of $\ER$)} \\
\Rightarrow & \forall i\leq k ~~\exists \pi'\in\Pi(\pi[i]) \\
   & ~~~~~(1) \forall j\geq 1 : K,\pi'[j] \models \beta \text{~~or~~~~(Claim)} \\
   & ~~~~~(2) \exists l\geq 1: K,\pi'[l] \models \alpha_1 ~\&~ 
                \text{~(ind. hypoth.)~} \\
   & ~~~~~~~~~       \forall q\leq l: K,\pi'[q] \models \beta ~\&~ K,\pi'[l] \models  \langle \alpha_2,\ldots,\alpha_m,\beta\rangle \\
\Rightarrow & \forall i\leq k : K,\pi[i] \models \beta \text{~~(since $\pi'[1]=\pi[i]$)} 
\end{align*}

``$\Leftarrow$'': 
\begin{align*}
& K,\pi[k] \models \langle \alpha_1,\ldots,\alpha_m,\beta\rangle \text{~and~} \forall i\leq k: K,\pi[i] \models \beta \\[1ex]
\Rightarrow & \forall i\leq k: K,\pi[i] \models \beta \text{~and~}  \exists \pi'\in\Pi(\pi[k]): \\
&  ~~~ (1) \forall j\geq 1: K,\pi'[j] \models \langle \alpha_2,\ldots,\alpha_m,\beta\rangle \text{~~or} \\
&  ~~~ (2) \exists j\geq 1: K,\pi'[j] \models \alpha_1 ~\&~ \forall q\leq j: K,\pi'[q] \models \langle \alpha_2,\ldots,\alpha_m,\beta\rangle \\[1ex]
\Rightarrow &  \exists \rho\in\Pi(w)  ~~(\text{where~}\rho=\pi[1]\cdots\pi[k](=\pi'[1])\pi'[2]\cdots) : \\
& ~~~~(1) \forall j\geq 1: K,\rho[j] \models \beta ~~(\text{using the above Claim}) ~~\text{or} \\
& ~~~~(2) \exists k'\geq k: K,\rho[k'] \models \alpha_1 ~\&~ 
               \forall i\leq k': K,\rho[i] \models \langle \alpha_2,\ldots,\alpha_m,\beta\rangle ~\text{(ind. hyp.)}~ \\[1ex]
\Rightarrow & \forall i\leq k: K,\pi[i] \models \langle \alpha_1,\ldots,\alpha_m,\beta\rangle
\end{align*}

\claimqed

The atomic right form of an $\{\ER\}$-formula is unique.

\begin{ownClaim}\label{claim:ER-in-LOGCFL-main}
Let $\varphi$ be an $\{\ER\}$-formula with atomic right form $(\alpha_1,\ldots,\alpha_m,\beta)$ and $m\geq 1$,
$K=(W,R,\xi)$ be a Kripke model, and $w\in W$.
Then $K,w \models \varphi$ if and only if
there exists a finite path $\pi$ through $(W,R)$ starting in $w$ with length $|\pi|\leq |W|+1$ such that
\begin{enumerate}
\item $K,\pi[i]\models \beta$ for all $i=1,2,\ldots,|\pi|$, and
\item\begin{enumerate}
\item $|\pi|=|W|+1$, or
\item 
\begin{itemize}
\item $K,\pi[|\pi|] \models \alpha_1$, and
\item $K,\pi[|\pi|] \models \alpha_2 \ER ( \alpha_3 \ER ( \cdots \ER (\alpha_m \ER \beta) \cdots ))$ \hfill 
                   (i.e. the formula with atomic right form $\langle\alpha_2,\ldots,\alpha_m,\beta\rangle)$.
\end{itemize}
\end{enumerate}
\end{enumerate}
\end{ownClaim}

$K,w\models \langle \alpha_1,\ldots,\alpha_m,\beta\rangle$ is defined as
\begin{align}
& \exists \pi\in\Pi(w) ~ \forall i\geq 1 : K,\pi[i]\models \langle \alpha_2,\ldots,\alpha_m,\beta\rangle ~~\text{or} \label{er-alg1} \\
& \exists \pi\in\Pi(w) ~ \exists k\geq 1 : K,\pi[k]\models \alpha_1 ~\&~ \forall j\leq k: K,\pi[j] \models \langle\alpha_2,\ldots,\alpha_m,\beta\rangle \label{er-alg2}
\end{align}
By Claim~\ref{claim:er-infinite-case} we get that (\ref{er-alg1}) is equivalent to the following.
\begin{align}
& \exists \pi\in\Pi(w) ~ \forall i\geq 1 : K,\pi[i]\models \beta \label{er-alg3} 
\end{align}
Since $\beta$ is an atom, 
(\ref{er-alg3}) is equivalent to
\begin{align*}
\text{there exists a finite path $\pi$ starting in $w$ with length $|\pi|=|W|+1$} \\
\text{such that $K,\pi[i]\models \beta$ for all $i=1,2,\ldots,|\pi|$.}
\end{align*}
This covers the first half (i.e. 2.a) of the claim.

Now consider (\ref{er-alg2}).
Using Claim~\ref{claim:er-finite-case} we get that (\ref{er-alg2}) is equivalent to
\begin{align}
& \exists \pi\in\Pi(w) ~ \exists k\geq 1 : \pi[k]\models \alpha_1 ~\&~ 
    \forall j\leq k: \pi[j] \models \beta ~\&~ \pi[k] \models \langle\alpha_2,\ldots,\alpha_m,\beta\rangle \label{er-alg4}
\end{align}
It is clear that if such a $k$ exists, then $k$ can be chosen to be $<|W|+1$.
This covers the second half (i.e. 2.b) of the claim.
\claimqed

Algorithm~\ref{algo:ER} implements this algorithm according to Claim~\ref{claim:ER-in-LOGCFL-main}.
It is easily seen to work in logarithmic space.
The stack is used for the recursive calls.
Since essentially every subformula causes one recursive call,
the algorithm runs in polynomial time.
Thus it is an $\LOGCFL$ algorithm.
\renewcommand*\lstlistingname{Algorithm}
\begin{lstlisting}[caption={LOGCFL machine that decides $\CTLMC{\ER}$},label=algo:ER,float, abovecaptionskip=-\medskipamount,mathescape]
Procedure: check
Input:     Kripke structure $K=(W,R,\xi)$, initial state $w_0\in W$, 
           formula $\phi$ with only $\ER$ operators.
Output:    true iff $K,w_0\models\phi$.

let $\langle \alpha_1,\ldots,\alpha_m,\beta\rangle$ be the atomic right form of $\phi$
if $\beta\not\in\xi(w_0)$ the return false
guess $\ell\leq|W|+1$
$i:=1$
$s:=w_0$
while $i\leq\ell$ do
  $s:= \text{ guessed successor of } s$
  if $\beta\not\in\xi(s)$ then return false
  $i:=i+1$
if $\ell=|W|+1$ or $m=0$ then return true
else return check$(K,s,\alpha_1)$ & check$(K,s,\langle \alpha_2,\ldots,\alpha_m,\beta\rangle)$
\end{lstlisting}
\end{proof}

%% file: Appendix-EG.tex
\subsection{$\EG$}

\begin{theorem}
$\CTLMC{\EG,\oplus}$ is $\P$-complete. 
\end{theorem}

\begin{figure}

\EGoplusfromGl

\caption{Kripke model $\KEGp$ obtained from the $\ASAGAPtwooutonein$ instance in Figure~\ref{fig:Bsp-ASAGAPveeonein}.}
\label{fig:ex-EGoplus}

\end{figure}

\begin{proof}
The upper bound $\P$ follows from \cite{clemsi86}.
For the lower bound---$\P$-hardness---we give a reduction from $\ASAGAPtwooutonein$.
Let $\langle G,s,T \rangle$ be an instance of $\ASAGAPtwooutonein$
with $G=(V,E)$ for $V=V_{\exists}\cup V_{\forall}$ with slices $V=V_0\cup\ldots\cup V_m$.
Let $G^{\flat}=(V^{\flat},E^{\flat})$ be the graph obtained from $G$
as described in Section \ref{subsubsec:Kripke}.
Using $G^{\flat}$, we construct a Kripke model $\KEGp=(V^{\flat},E^{\flat},\xi)$ with assignment $\xi$ as follows 
(see Figure~\ref{fig:ex-EGoplus} for an example).
\begin{enumerate}
\item $s_i$ is assigned to all nodes in  $V^{\flat}_i$.

\item $\hat{s_i}$ is assigned to all nodes $\hat{V}_i$.
      
\item $t$ is assigned to all nodes in $T$.
\end{enumerate}

The formulas $\varphi_i$ ($i=m,m-1,\ldots,0$) are inductively defined as follows.

$$
\varphi_i = 
\begin{cases}
t, & \text{ if } i=m \\
\EG(s_i \oplus s_{i+2} \oplus \hat{s}_i \oplus \hat{s}_{i+1} \oplus \varphi_{i+1}), & \text{ if } i<m 
\end{cases}
$$

We have the following easy-to-see properties of the model $\KEGp$ and the formulas $\varphi_i$.
\begin{ownClaim}\label{claim:EG-easy}
\begin{enumerate}
\item\label{easy1}
 For all $i\leq m$, all nodes $w\in V^{\flat}_i$, and all $j>i$  holds $\KEGp,w \nmodels \varphi_{j}$.

\item\label{easy2} 
 For all $i\leq m$ and all nodes $z\in \hat{V}_i$ holds $\KEGp,z \nmodels \varphi_i$.

\item\label{easy3} 
 For all $i\leq m$ and all $u\in \hat{V}_{i}$ with $(u,u)\in E^{\flat}$ holds $\KEGp,u \models \varphi_{i-1}$.

\end{enumerate}
\end{ownClaim}

We sketch the proof.
For \ref{easy1}: $\KEGp,w \nmodels \varphi_{j}$ since no atoms that appear in $\varphi_j$ are assigned to node $w$ in slice $i<j$.

For \ref{easy2}:
By the definition of $\xi$ we have $z\in \xi(s_i)$ and $z\in \xi(\hat{s}_i)$,
and $z\not\in \xi(s_{i+2})$ and $z\not\in \xi(\hat{s}_{i+1})$.
With case \ref{easy1} we also have $\KEGp,z\nmodels\varphi_{i+1}$.
Thus $\KEGp,z\nmodels s_i \oplus s_{i+2} \oplus \hat{s}_i \oplus \hat{s}_{i+1} \oplus \varphi_{i+1}$,
and consequently $\KEGp,z\nmodels \varphi_i$.

For \ref{easy3}:
 For $\varphi_{i-1} = \EG(s_{i-1} \oplus s_{i+1} \oplus \hat{s}_{i-1} \oplus \hat{s}_{i} \oplus \varphi_{i})$,
 we have that $\hat{s}_i\in \xi(u)$ and $s_{i-1}, s_{i+1}, \hat{s}_{i-1}\not\in\xi(u)$.
 From \ref{easy2} we get $\KEGp,u\nmodels \varphi_{i}$.
 Thus, $\KEGp,u\models s_{i-1} \oplus s_{i+1} \oplus \hat{s}_{i-1} \oplus \hat{s}_{i} \oplus \varphi_{i}$.
 Since all infinite paths $\pi\in\Pi(u)$ only loop through $u$ (e.g. $\pi[k]=u$ for all $k\geq 1$),
 it follows that $\KEGp,u\models \varphi_{i-1}$.
\claimqed

\begin{ownClaim}\label{claim:upper}
For all $i\leq m$, all nodes $w\in V^{\flat}_i$, and all $j\leq i-2$  holds: $\KEGp,w \nmodels \varphi_{j}$.
\end{ownClaim}

The proof is by induction on $i$. 

\begin{itemize}
\item Base case $i=m$. Consider $w\in V^{\flat}_m$.
Notice that every infinite path $\pi\in\Pi(w)$ eventually loops in a node  $u_w\in\hat{V}_m$ with $(u_w,u_w)\in E^{\flat}$.
Since $K,v \models \EG\alpha$ if and only if $K,v \models \alpha$ and $K,v'\models \EG\alpha$ for some successor $v'$ of $v$,
it suffices to show that $\KEGp, u_w \nmodels \varphi_{j}$.
We proceed by induction on $j$, where $j=m-2$ is the base case.
Since $\KEGp, u_w \models \varphi_{m-1}$ (Claim \ref{claim:EG-easy}(\ref{easy3})),
it follows that $\KEGp, u_w \nmodels s_{m-2} \oplus s_{m} \oplus \hat{s}_{m-2} \oplus \hat{s}_{m-1} \oplus \varphi_{m-1}$,
and thus $\KEGp, u_w \nmodels \varphi_{m-2}$.
(Generally, a formula $\EG\alpha$ is satisfied in a node $u$
if and only if $u\models\alpha$ and $v\models\EG\alpha$ for some successor $v$ of $u$.)

For $j<m-2$, we have the inductive hypothesis $\KEGp, u_w \nmodels \varphi_{j+1}$.
Since $\KEGp, u_w \nmodels s_j \oplus s_{j+2} \oplus \hat{s}_j \oplus \hat{s}_{j+1}$,
it follows that $\KEGp, u_w \nmodels s_j \oplus s_{j+2} \oplus \hat{s}_j \oplus \hat{s}_{j+1} \oplus \varphi_{j+1}$
and thus $\KEGp, u_w \nmodels \varphi_j$.

\item Inductive step $i<m$. Consider node $w\in V^{\flat}_i$.
Again we proceed by induction on $j$.
\begin{itemize}
\item Base case $j=i-2$ for $\varphi_j = \EG(s_j \oplus s_{j+2} \oplus \hat{s}_j \oplus \hat{s}_{j+1} \oplus \varphi_{j+1})$.
Since no states in slices $\geq i$ satisfy $s_j$, $\hat{s}_j$, and $\hat{s}_{j+1}$, and $\KEGp,w\models s_{j+2}(=s_i)$,
it follows that $\KEGp,w\models \varphi_{j}$ iff $\KEGp,w\models \EG(s_{j+2}\oplus\varphi_{j+1})$.
By inductive hypothesis we have $\KEGp,v\nmodels\varphi_{j+1}$ for all $v\in V^{\flat}_{i+1}$.
These nodes $v$ do not satisfy $s_{j+2}$.
Therefore $\KEGp,w\models \EG(s_{j+2}\oplus\varphi_{j+1})$ only holds,
if it is witnessed by a path that stays in slice $V^{\flat}_i$.
This path eventually loops in a node $u_w\in V_{i}$ with $(u_w,u_w)\in E^{\flat}$.
By Claim \ref{claim:EG-easy}(\ref{easy3}) we know $\KEGp, u_w\models\varphi_{j+1}$ (since $j+1=i-1$).
Consider $\varphi_{j} = \EG(s_{j} \oplus s_{j+2} \oplus \hat{s}_{j} \oplus \hat{s}_{j+1} \oplus \varphi_{j+1})$.
We have that $s_{j+2}$ and $\varphi_{j+1}$ are the only ``parts'' of $\varphi_j$ that are satisfied in $u_w$.
Thus $\KEGp,u_w\nmodels s_{j} \oplus s_{j+2} \oplus \hat{s}_{j} \oplus \hat{s}_{j+1} \oplus \varphi_{j+1}$,
and therefore $\KEGp,u_w\nmodels \varphi_{j}$.

Since every path from $w$ that stays in slice $i$ ends in such a node $u_w$,
we get that $\KEGp,w\nmodels \varphi_{j}$.
\item Inductive step $j<i-2$. 
Consider $\varphi_{j} = \EG(s_{j} \oplus s_{j+2} \oplus \hat{s}_{j} \oplus \hat{s}_{j+1} \oplus \varphi_{j+1})$.
By inductive hypothesis we know $\KEGp,w\nmodels \varphi_{j+1}$.
Moreover, $s_j, s_{j+2}, \hat{s}_j, \hat{s}_{j+1}\not\in\xi(w)$.
Therefore $\KEGp,w\nmodels \varphi_j$. \claimqed 
\end{itemize}
\end{itemize}

\begin{ownClaim}\label{claim:EG-oplus-main}
For every $i\leq m$ and every $w\in V_i$ holds:
$\KEGp,w\models \varphi_i$ if and only if $\apath_G(w,T)$.
\end{ownClaim}

The proof proceeds by induction on $i$.
The base case $i=m$ is straightforward.

For the inductive step $i<m$, consider $w\in V_i$.

$\KEGp,w\models \EG(s_{i} \oplus s_{i+2} \oplus \hat{s}_{i} \oplus \hat{s}_{i+1} \oplus \varphi_{i+1}) (=\varphi_i)$
if and only if there exists an infinite path $\pi\in\Pi(w)$ such that
$\KEGp,\pi[j] \models s_{i} \oplus s_{i+2} \oplus \hat{s}_{i} \oplus \hat{s}_{i+1} \oplus \varphi_{i+1}$ $(=:\alpha_i)$ for all $j$.
Notice that this is equivalent to $\KEGp,\pi[j] \models \varphi_i$ for all $j$.
 
Assume that such a $\pi$ exists.
Since $\KEGp,w\models s_i$ and $\KEGp,w\nmodels s_{i+2}, \hat{s}_i, \hat{s}_{i+1}, \varphi_{i+1}$, it holds that $\pi[1]\models\alpha_i$.
For the ``right neighbour'' $v\in \hat{V}_i$ of $w$
holds $\KEGp,v\nmodels\varphi_i$ (Claim \ref{claim:EG-easy}(\ref{easy2})).
This means that $\pi[2]\in V_{i+1}$.
Then $s_{i+1}\in\xi(\pi[2])$ and $s_i, s_{i+2}, \hat{s}_i, \hat{s}_{i+1}\not\in\xi(\pi[2])$.
Therefore, $\KEGp,\pi[2]\models \alpha_i$ if and only if $\KEGp,\pi[2]\models \varphi_{i+1}$.

Since no node in layer $V^{\flat}_{i+2}$ satisfies $\varphi_i$ (Claim \ref{claim:upper}),
we conclude that $\pi[3],$ $\pi[4],\ldots$ must be in slice $V^{\flat}_{i+1}$.
If $(\pi[3],\pi[3])\in E^{\flat}$, we are done as $\KEGp,\pi[3]\models \varphi_i$ by Claim \ref{claim:EG-easy}(\ref{easy3}).
Otherwise,
$\pi[3]$ is a node in $\hat{V}_{i+1}$.
By Claim \ref{claim:EG-easy}(\ref{easy2}) we have $\KEGp,\pi[3]\nmodels\varphi_{i+1}$.
Since $s_i, s_{i+2}, \hat{s_i} \not\in\xi(\pi[3])$ and $\hat{s}_{i+1}\in\xi(\pi[3])$,
we get $\KEGp,\pi[3]\models \alpha_i$.

Now, $\pi[q]$ for even $q\geq 4$ can be dealt like $\pi[2]$,
and $\pi[r]$ for odd $r\geq 5$ can be dealt like $\pi[3]$.
Let $\pi[1],\pi[2],\ldots,\pi[5]$ be the finite prefix of $\pi$ that ends in the node through which $\pi$ eventually loops.
We have seen that $\KEGp,w\models \varphi_i$ if and only if $\KEGp,\pi[2]\models \varphi_{i+1}$ and $\KEGp,\pi[4]\models\varphi_{i+1}$.
By the inductive hypothesis this is equivalent to $\apath_G(\pi[2],T)$ and $\apath_G(\pi[4],T)$.
Since $\pi[2]$ and $\pi[4]$ are all the successors of $\pi[1]=w$ in $G$,
the latter is equivalent to $\apath_G(w,T)$.
\claimqed

With Claim~\ref{claim:EG-oplus-main}
we get that $\langle G,s,T\rangle \in \ASAGAPtwooutonein$ if and only if $\KEGp,s\models\varphi_0$.
The $\CTLMC{\EG,\oplus}$ instance $\langle \KEGp,s, \varphi_{0}\rangle$
can be computed in space logarithmic in the size of $G$.
Thus $\ASAGAPtwooutonein$ logspace reduces to $\CTLMC{\EG,\oplus}$.
\end{proof}


\begin{lemma}\label{lemma:EG-NL-hard}
$\CTLMC{\EG}$ is $\NL$-hard.
\end{lemma}

\begin{proof}
We give a logspace reduction from the $\NL$-complete  graph accessibility problem. 
Let $(G,s,t)$ be the given $\GAP$ instance with $G=(V,E)$.
Let $V'=\{(u,i)\mid u\in V, 1\leq i\leq |V|\}$ be a set consisting of $|V|$ copies of every node in $|V|$,
and $E'$ be a set of edges on $V'$ similar to $E$,
such that an edge $(u,v)\in E$ leads to edges from the $i$th copy of $u$ to the $(i+1)$st of $v$,
plus reflexive edges for all $|V|$th copies,
i.e., $E'=\{((u,i),(v,i+1))\mid (u,v)\in E, 1\leq i < |V|\} \cup \{((u,|V|),(u,|V|))\mid u\in V\}$.
The assignment $\xi$ assigns $a$ to all nodes $(u,i)\in V'$ with $i<|V|$ or $u=t$.
Let $M=(V',E',\xi)$ be a Kripke model.
It is clear that $G$ has an $s\text-t$-path if and only if $M,(s,1)\models\EG~a$.
\end{proof}

\begin{lemma}\label{lem:EG-AF-in-NL}
$\CTLMC{\EG}$ and $\CTLMC{\AF}$ are in $\NL$.
\end{lemma}

\begin{proof}
First note that $\EG\cdots\EG p \equiv \EG p$.

The algorithm for $\CTLMC{\EG}$ gets input $\langle (W,R,\xi),w_0,\EG^k p\rangle$.
If $k=0$ (i.e.~the formula to check equals $p$),
it checks whether $w_0\in\xi(p)$ and decides accordingly.
If $k>0$, the algorithm must verify whether $(W,R)$ has an infinite paths starting in $w_0$
on which $p$ is satisfied in every point.
The existence of such a path is equivalent to the existence
of two paths $w_0=v_1,v_2,\ldots,v_m$ and $v_m=u_1,u_2,\ldots,u_q=v_m$
for some $m,q\leq |W|$ such that $p$ is satisfied by all $v_i$ and $u_i$.
Both paths together form an ultimately periodic infinite path that is searched for.
The algorithm first guesses $v_m$ and $q$,
and then stepwise guesses the paths and verifies that $p$ is satisfied always.
This is clearly an $\NL$-algorithm.

Since $K,w_0\models \AF p$ iff $K,w_0\nmodels \EG \neg p$,
the above $\CTLMC{\EG}$-algorithm can be used to decide $\overline{\CTLMC{\AF}}$.
Since $\NL$ is closed under complement, $\CTLMC{\AF}$ is in $\NL$, too.
\end{proof}

\begin{lemma}\label{lemma:EG-and-NL}
$\CTLMC{\EG,\wedge}$ is in $\NL$.
\end{lemma}

\begin{proof}
We first notice that $\EG(\alpha\wedge\EG\,\beta)\equiv\EG(\alpha\wedge\beta)$.
(1) If $K,w\models \EG(\alpha\wedge\EG\,\beta)$,
then there exists a path starting in $w$ on which everywhere $\alpha\wedge\beta$ is satisfied.
(2) If $K,w\models \EG(\alpha\wedge\beta)$,
then there exists a path $\pi$ starting in $w$ on which everywhere $\alpha\wedge\beta$ is satisfied.
Then $K,\pi[m]\models \EG\beta$ (witnessed by $\pi^m$) for every $m$.
Therefore  $K,w\models \EG(\alpha\wedge\EG\beta)$.

Due to $\EG\EG\alpha\equiv\EG\alpha$ and the above equivalence,
every $\{\EG,\wedge\}$-formula can be transformed 
to an equivalent formula of the form $\alpha \wedge \bigwedge_{\ell=1,2,\ldots,k} \EG~\beta_{\ell}$,
where $\alpha$ and all $\beta_{\ell}$ are conjunctions of atoms.
The satisfaction $K,s\models\EG~\beta_{\ell}$
can be checked nondeterministally within logspace by guessing the relevant prefix of a looping infinite path
that satisfies $\beta_{\ell}$ in every node.
Doing this for all $\EG$-subformulas yields an $\NL$-algorithm for $\CTLMC{\EG,\wedge}$.
\end{proof}

\begin{lemma}\label{EG-vee-NL}
$\CTLMC{\EG,\vee}$ is in $\NL$.
\end{lemma}

\begin{proof}
Every $\{\EG,\vee\}$-formula can be transformed into an equivalent
formula of the form $\alpha \vee \bigvee_{\ell=1,2,\ldots,k} \EG~\beta_{\ell}$ $(\ast)$,
where $\alpha$ is a disjunction of atoms and every $\beta_{\ell}$
is a formula of the form $(\ast)$ (for $k=0$, such a formula is a disjunction of atoms).\medskip

\begin{ownClaim}
Let $K=(W,R,\xi)$ be a Kripke model,
$\alpha$ be a disjunction of atoms, and $\beta_{\ell}$ be formulas of the form $(\ast)$.
Then $K,s\models  \EG(\alpha \vee \bigvee_{\ell=1,2,\ldots,k} \EG~\beta_{\ell})$
if and only if
\begin{enumerate}
\item\label{fall1} there is a path $v_1,\ldots,v_m$ through $K$ starting in $s$ and of length $m=|W|+1$
      such that $K,v_i\models\alpha$ for $i=1,2,\ldots,m$, or
\item\label{fall2} there is a path $v_1,\ldots,v_m$ through $K$ starting in $s$ and of length $1\leq m\leq|W|+1$
      such that $K,v_i\models\alpha$ for $i=1,2,\ldots,m-1$ and $K,v_m\models \EG~\beta_{q}$ for some $q$.
\end{enumerate}
\end{ownClaim}
\begin{proof}
The implication from left to right is straightforward.
Consider the other proof direction.
If \ref{fall1} happens, then $R$ contains an edge from $v_m$ to some predecessor on the path.
Using this loop we get an infinite path that satisfies $\alpha$ on every of its nodes.
If \ref{fall2} happens, then let $u_1,u_2,\ldots$ be the infinite path starting with $v_m=u_1$
such that $K,u_i\models\beta_q$ for all $i\geq 1$.
Then $K,u_i\models\EG\,\beta_q$ for all $i\geq 1$.
Consequently, on every node of the infinite path $v_1(=s),\ldots,v_m(=u_1),u_2,\ldots$
the formula $\alpha \vee \bigvee_{\ell=1,2,\ldots,k} \EG~\beta_{\ell}$ is satisfied.
Therefore $K,s\models \EG(\alpha \vee \bigvee_{\ell=1,2,\ldots,k} \EG~\beta_{\ell})$.
\claimqed 
\end{proof}

Using this claim,
an $\NL$-algorithm can proceed as follows.
On input $K,s,\alpha \vee \bigvee_{\ell=1,2,\ldots,k} \EG~\beta_{\ell}$,
it accepts if $K,s\models\alpha$.
Otherwise, it guesses an $i$ and goes to check $K,s\models\EG~\beta_{i}$
where $\beta_i=\alpha' \vee \bigvee_{\ell=1,2,\ldots,k} \EG~\beta'_{\ell})$.
For this, 
it guesses which of the two cases of the Claim has to be fulfilled.
Case 1 can be verified straightforwardly.
For case 2, it guesses the relevant $m$ and $q$,
guesses $v_i$ and checks that $K,v_i\models \alpha'$ for $i=1,2,\ldots,m-1$ 
and eventually recursively checks whether $K,v_m\models \EG~\beta'_{j}$ for some $j$.
Since this is a tail recursion whose depth is bounded by the depth of the input formula,
it can be performed nondeterministically within logspace. 
\end{proof}

\begin{lemma}\label{lemma:EG-neg-NL}
$\CTLMC{\EG,\neg}$ is in $\NL$.
\end{lemma}

\begin{proof}
The following equivalences hold for $\EG$ and its dual $\AF$.
\begin{enumerate}
\item $\EG~\EG~ \alpha \equiv \EG~\alpha$
\item $\AF~\AF\alpha \equiv \AF~\alpha$
\item\label{EGAFEG} $\EG~\AF~\EG~\alpha \equiv ~\AF~\EG~\alpha$
\item\label{AFEGAF} $\AF~\EG~\AF~\alpha \equiv \EG~\AF~\alpha$
\end{enumerate}

Proof of \ref{EGAFEG}:
If $K,w\models\EG~\AF~\EG~\alpha$, then clearly $K,w\models\AF~\EG~\alpha$.
For the other direction, assume $K,w\models\AF~\EG~\alpha$.
Take some $\pi\in\Pi(w)$.
Then $K,\pi[k] \models \EG\alpha$ for some ``smallest'' $k$
with $K,\pi[i] \nmodels \EG\alpha$ for $i=1,2,\ldots,k-1$.
Since $K,\pi[1]\models \AF~\EG~\alpha$,
it follows that $K,\pi[i] \models \AF~\EG~\alpha$ for $i=1,2,\ldots,k-1$.
Moreover, let $\rho$ be a path that witnesses $K,\pi[k]\models \EG \alpha$.
Then $\rho^j$ witnesses $K,\rho[j] \models \EG\,\alpha$ for all $j\geq 1$,
and from $K,\rho[j] \models \EG\,\alpha$ follows $K,\rho[j] \models \AF\,\EG\,\alpha$.
Concluding we have for the infinite path $\lambda=(\pi[1](=w),\pi[2],\ldots,\pi[k-1],\rho[1](=\pi[k]),\rho[2],\ldots)$
that $K,\lambda[i] \models \AF\,\EG\,\alpha$ for all $i$,
what means that $\lambda$ is a witness for $K,w\models \EG\,\AF\,\EG\,\alpha$.

The proof of \ref{AFEGAF} follows from \ref{EGAFEG} by the duality of $\AF$ and $\EG$. 

These equivalences yield that every $\{\EG,\neg\}$-formula with atom $p$
is equivalent to $\EG~\AF~p$  or $\AF~\EG~p$, 
or to $\EG p$ or $\AF p$, or to $p$,
or to one of these formulas where $p$ is replaced by $\neg p$.
For a given $\{\EG,\neg\}$-formula it can be checked in logarithmic space
to which of these cases the formula belongs.

We first describe an algorithm for the $\EG~\AF~p$ case.

The algorithm gets input $\langle (W,R,\xi),w_0,\EG\AF p\rangle$.
It must verify whether $(W,R)$ has an infinite paths starting in $w_0$
on which $\AF p$ is satisfied in every point.
The existence of such a path is equivalent to the existence
of two paths $w_0=v_1,v_2,\ldots,v_m$ and $v_m=u_1,u_2,\ldots,u_q=v_m$
for some $m,q\leq |W|$ such that $\AF p$ is satisfied by all $v_i$ and $u_i$.
Both paths together form an ultimately periodic infinite path that is searched for.
The algorithm first guesses $v_m$ and $q$,
and then stepwise guesses the paths and verifies that $\AF p$ is satisfied always.
This is done by guessing the next $v_i$ (resp.~$u_i$),
and then starting the (slightly modified) $\NL$-algorithm for $\CTLMC{\AF}$
with input $\langle (W,R,\xi), v_i, \AF p \rangle$.
If it reaches an accepting configuration, then the next $v_i$ (resp.~$u_i$) is guessed etc.

This also yields an $\NL$-algorithm.

The algorithms for the other cases are constructed in the same way.
Since $\NL$ is closed under complement,
all algorithms are $\NL$-algorithms.
\end{proof}

%% file: Appendix-EF.tex
\subsection{$\EF$}

\begin{theorem}\label{lemma:EF-or-is-NL-complete}
$\CTLMC{\EF}$ and $\CTLMC{\EF,\vee}$ are $\NL$-complete. 
\end{theorem}

\begin{proof}
It suffices to show $\NL$-hardness of  $\CTLMC{\EF}$
and containment in $\NL$ of $\CTLMC{\EF,\vee}$.

$\NL$-hardness of  $\CTLMC{\EF}$ follows by a reduction from the directed graph accessability problem as follows.
Let $\langle (V,E),s,t \rangle$ be an instance of the graph accessability problem---i.e. we
want to decide whether graph $(V,E)$ has an $s$-$t$-path.
Let $\hat{E}$ be the reflexive closure of $E$.
Then $(V,\hat{E})$ is a total graph, and it has an $s$-$t$-path if and only if $(V,E)$ has some.
Define the assignment $\xi$ as $\xi(t)=\{p\}$ and $\xi(w)=\emptyset$ for $w\not=t$.
Then $(V,E)$ has an $s$-$t$-path if and only if $(V,\hat{E},\xi),s \models \EF~p$.

For $\CTLMC{\EF,\vee}\in\NL$,
note that $\EF(\alpha\vee\EF\beta)\equiv \EF(\alpha\vee\beta)$
and $\EF\alpha\vee\EF\beta \equiv \EF(\alpha\vee\beta)$.
Thus, every $\{\EF,\vee\}$-formula can be transformed into an equivalent
formula of the form $\alpha \vee \EF~\beta$,
where $\alpha$ and $\beta$ are disjunctions of atoms.
This transformation can be done in logarithmic space.
The $\NL$ algorithm on input $\langle K,w_0,\phi \rangle$
verifies whether $K,w_0\models \alpha$ or guesses a reachable $v$ 
and verifies $K,v\models\beta$.
\end{proof}

\begin{theorem}\label{thm:EFneg}
$\CTLMC{\EF,\neg}$ is $\NL$-complete. 
\end{theorem}

\begin{proof}
$\NL$-hardness follows from that of $\CTLMC{\EF}$ (Theorem~\ref{lemma:EF-or-is-NL-complete}).

Every $\{\EF,\neg\}$-formula can be rewritten as 
a formula with $\EF$s and $\AG$s followed by a literal $p$ or $\neg p$.
It is clear that $\EF \EF \alpha \equiv \EF \alpha$ and $\AG \AG \alpha \equiv \AG \alpha$.
Thus every such formula can be rewritten as one having a prefix of alternating $\EF$s and $\AG$s.
With the equivalences of the following claim we can reduce this prefix to length $\leq 3$.

\begin{ownClaim}\label{claim:EFAG}
Let $K$ be a Kripke model, $w$ be a node of $K$, and $\alpha$ be a $\CTL$-formula.
\begin{enumerate}
\item\label{EFEF} $K,w\models \EF\EF\alpha$ if and only if $K,w\models\EF\alpha$.
\item\label{AGAG} $K,w\models \AG \AG \alpha$ if and only if $K,w\models \AG \alpha$.
\item\label{EFAG} $K,w\models \EF \AG \EF \AG  \alpha$ if and only if $K,w\models \EF \AG  \alpha$.
\end{enumerate}
\end{ownClaim}

(\ref{EFEF}) and (\ref{AGAG}) are straightforward.
For (\ref{EFEF}), notice that $K,w\models\alpha$ implies $K,w\models\EF \alpha$.
For (\ref{EFAG}), we consider both proof directions separately.
\vspace*{1ex}

$
\begin{array}{@{}rrlr}
\multicolumn{4}{l}{K,w\models \EF \AG \EF \AG  \alpha} \\
 & \Rightarrow &
\exists \pi\in \Pi(w) ~ \exists k\geq 1 ~ \forall \rho\in\Pi(\pi[k]) ~\\
&&~~\forall j\geq 1 :  K,\rho[j] \models \EF \AG  \alpha 
& \text{(semantics \ldots)} \\
& \Rightarrow & \exists \pi\in \Pi(w) ~ \exists k\geq 1 ~ \forall \rho\in\Pi(\pi[k]) :  K,\rho[1] \models \EF \AG  \alpha 
& \text{(take $j=1$)} \\
& \Rightarrow & \exists \pi\in \Pi(w) ~\exists k\geq 1:  K,\pi[k] \models \EF \AG  \alpha & \text{($\rho[1]=\pi[k]$)} \\
& \Rightarrow & K,w \models \EF \EF \AG  \alpha & \text{(semantics of $\EF$)} \\
& \Rightarrow & K,w \models \EF \AG  \alpha     & \text{(part (\ref{EFEF}))}
\end{array}
$
\vspace*{2ex}

For the other direction,
we use (\ref{AGAG}) and the fact that $K,w\models\beta$ implies $K,w\models\EF \beta$.
\vspace*{1ex}

$
\begin{array}{@{}rrlr}
\multicolumn{4}{l}{K,w \models \EF \AG  \alpha} \\
& \Rightarrow &
K,w \models \EF \AG \AG \alpha \\
& \Rightarrow &
K,w \models \EF \AG \EF \AG \alpha 
\end{array}
$
\vspace*{-2ex}

\claimqed
\vspace*{1ex}

By Claim \ref{claim:EFAG} follows
that every formula in the $\{\EF,\neg\}$-fragment has an equivalent formula
in the $\{\EF,\AG\}$-fragment with atomic negation,
whith a prefix of at most three temporal operators.
Since $\CTLMC{\EF}$ and $\CTLMC{\AG}$ are in $\NL$
(follows from Theorem~\ref{lemma:EF-or-is-NL-complete} and the closure of $\NL$ under complement), similar as in the proof of Lemma~\ref{lemma:EG-neg-NL},
an $\NL$-algorithm can be composed that combines the $\CTLMC{\EF}$ and $\CTLMC{\AG}$ algorithms
in order to evaluate the bounded number of alternations of temporal operators.
\end{proof}

\begin{theorem}
$\CTLMC{\EF,\oplus}$ is $\AC1$-hard. 
\end{theorem}

\begin{proof}
We give a reduction from the $\AC1$-complete problem $\ASAGAP_{\log}$.
Let $\langle G,s,T \rangle$ be an instance of $\ASAGAP_{\log}$,
where $V=V_{\forall} \cup V_{\exists}$
consists of slices $V = V_0 \cup V_1 \cup \ldots \cup V_{\ell}$.
W.l.o.g. we assume that $V_{\ell} \subseteq V_{\forall}$. 

Next we describe the construction of a Kripke model $\KEFo$ that bases on $G$.
In order to ease the readability of the proof,
we prefer to use the indices of the slices in reverse order.
Let $W'_0=V_{\ell}$ (the slice with nodes without successors),
$W'_1 = V_{\ell-1}, W'_2=V_{\ell-2}, \ldots, W'_{\ell}=V_0$.
By the above convention, $W'_0$ consists of $\forall$-nodes.
Thus $V_{\forall} = \bigcup_{i \text{ even}} W'_i$ and $V_{\exists} =   \bigcup_{i \text{ odd}} W'_i$.
Eventually, we add $2(\ell+1)$ new nodes and define $W_i = W'_i \cup \{a_i, b_i\}$ for $i=0,1,\ldots,\ell$.
We will call each $W_i$ as \emph{layer $i$},
and $W=\bigcup\limits_{i=0}^{\ell} W_i$ is the set of nodes of $\KEFo$.

Next we consider the edges.
We take all edges from $E$, and add loops $(u,u)$ for all $u\in W_0$.
The new $a_i$ nodes form a path $\{(a_i,a_{i-1}) \mid i=\ell,\ell-1,\ldots,1\}$,
and the new $b_i$ nodes form a path $E_b=\{(b_i,b_{i-1}) \mid i=\ell,\ell-1,\ldots,1\}$.
Moreover, for odd $i$ every node $u\in W_i$ has an edge $(u,a_{i-1})$ to $a_{i-1}$,
and for even $i\geq 2$  every node $u\in W_i$ has an edge $(u,b_{i-1})$ to $b_{i-1}$.
Let $E'$ denote this set of edges.

We complete the description of $\KEFo$ with the assignment $\xi$.
It marks each layer $W_i$ with an individual atom $z_i$.
Moreover, the nodes in $T$ and $b_0$ are marked with $t$.
$$
\xi(w)=\begin{cases}
        \{z_0,t\}, & \text{ if $w\in W_0 \cap (T\cup\{b_0\})$} \\
        \{z_i\},   & \text{ if $w\in W_i \cap \overline{T\cup\{b_0\}}$}
       \end{cases}
$$
The Kripke model $\KEFo$ constructed from $G$ is defined as $\KEFo=(W,E',\xi)$.
Figure~\ref{fig:EFoplus} shows an example for the construction.
(This example does not have logarithmic depth, but gives
a good insight into the construction.)

\begin{figure}[ht]
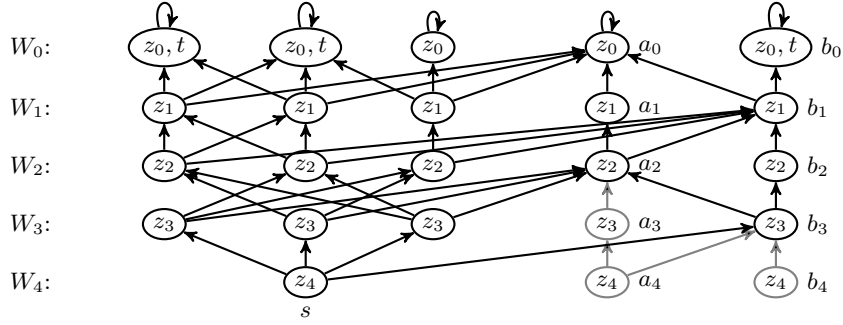


\centering

\EFoplusACeinshart

\caption{Kripke model $\KEFo$ constructed from an alternating graph.}
\label{fig:EFoplus}
\end{figure}

For $i=0,1,2,\ldots,\ell$, we inductively define formulas $\varphi_i$ as follows.
Here we use $\AG \alpha$ as abbreviation for $s \oplus \EF (s\oplus \alpha)$
for a new atom $s$ that is satisfied in every node of the Kripke model.
Under this condition, $s \oplus \EF (s\oplus \alpha) \equiv \neg \EF \neg \alpha \equiv \AG \alpha$.
$$
\varphi_i = 
\begin{cases}
t, & \text{ if $i=0$} \\
\EF\Big(\bigoplus\limits_{z\in\alpha(i)} z ~\oplus~ \bigoplus\limits_{j=0}^{i-1} \varphi_j \Big), & 
                   \text{ if $i>0$ and odd (layer $i$: $\exists$-nodes)} \\
\AG\Big(\bigoplus\limits_{z\in\alpha(i)} z ~\oplus \bigoplus\limits_{j=0}^{i-1} \varphi_j \Big), & 
                   \text{ if $i>0$ and even (layer $i$: $\forall$-nodes)} 
\end{cases}
$$
where each $\alpha(i)$ is a subset of $\{z_0,z_1,\ldots,z_i\}$ defined as follows.
Let $\# A$ denote the number of elements of the set $A$.
\begin{itemize}
\item for odd $i$ and $j< i-1$ :

$z_j\in\alpha(i)$ ~~~iff 

~~~$\#\{ m \in\{0,1,2,\ldots,j-1\} \mid m \text{ odd} \} + \#\{ m\in\{j+2,\ldots,i-1\} \mid m \text{ even}\}$ is odd

\item for odd $i$ and $j\in\{i-1,i\}$ :

$z_j\in\alpha(i)$ ~~~iff 
~~~$\#\{m\in\{0,1,2,\ldots,i-2\} \mid m \text{ odd}\}$ is odd

\item for even $i$ and $j< i-1$ :

$z_j\in\alpha(i)$ ~~~iff 

~~~$\#\{ m \in\{0,1,2,\ldots,j-1\} \mid m \text{ odd} \} + \#\{ m\in\{j+2,\ldots,i-1\} \mid m \text{ even}\}$ is even

\item for even $i$ and $j\in\{i-1,i\}$ :

$z_j\in\alpha(i)$ ~~~iff 
~~~$\#\{m\in\{0,1,2,\ldots,i-2\} \mid m \text{ odd}\}$ is odd
\end{itemize}
As examples, we write down $\varphi_0$, $\varphi_1$, $\varphi_2$, and $\varphi_3$.
\begin{itemize}
\item $\varphi_0=t$.
\item
For $\varphi_1$: $z_0,z_1\not\in\alpha(1)$ since $\#\{m\in\emptyset\mid m\text{ odd}\}=0$ is even.

Thus, $\varphi_1=\EF t$.
\item
For $\varphi_2$: $z_0\in\alpha(2)$ since $\#\{m\in\emptyset\mid m\text{ odd}\}+\#\{m\in\emptyset\mid m\text{ even}\}=0$ is even.
$z_1,z_2\not\in\alpha(2)$ since $\#\{m\in\{0\}\mid m\text{ odd}\}=0$ is even.

Thus $\varphi_2=\AG(z_0 \oplus t \oplus \EF t)$.
\item
For $\varphi_3$: $z_0\in\alpha(3)$ since $\#\{m\in\emptyset\mid m\text{ odd}\}+\#\{m\in\{2\}\mid m\text{ even}\}=1$ is odd.
$z_1\in\alpha(3)$ since $\#\{m\in\{3\}\mid m\text{ odd}\}+\#\{m\in\emptyset\mid m\text{ even}\}=1$ is odd.
$z_2,z_3\in\alpha(3)$ since $\#\{m\in\{0,1\}\mid m\text{ odd}\}=1$ is odd.

Thus $\varphi_3=\EF(z_0\oplus z_1 \oplus z_2\oplus z_3 \oplus t \oplus \EF t \oplus \AG(z_0 \oplus t \oplus \EF t))$.
\end{itemize}

The following claim contains the crucial properties of Kripke model $\KEFo$ and the formulas $\varphi_i$.

\begin{ownClaim}\label{claim:EFoplus-props}
For every $j=0,1,\ldots,\ell$ and every node $w_j\in W_j$ the following holds.
\begin{description}
\item[(A)] $\KEFo,w_j \models \varphi_j$ if and only if $\KEFo,w_j\models \varphi_{j+1}$.
\item[(B)] For all $i\geq j+2$ holds $\KEFo,w_j \models \varphi_i$ if and only if $i$ is even.
\item[(C)] $\KEFo,b_j\models \varphi_j$ and $\KEFo,a_j\nmodels\varphi_j$.
\item[(D)] For $w_j\in V\cap W_j$: $\KEFo,w_j\models \varphi_j$ if and only if $\apath_G(w_j,T)$.
\end{description}
\end{ownClaim}

\begin{proof}
 Throughout the proof,
we will use the following straightforward connection between sums of $z_i$ and their satisfaction in different layers.
By the construction of the Kripke model $\KEFo$, each $z_j$ is satisfied exactly in nodes of layer $W_j$.
Therefore  for all $i\geq j$ and 
\begin{align}
\text{for every $w_j\in W_j$ holds:}~~z_j\in\alpha(i) \text{~~if and only if~~} \KEFo,w_j\models \bigoplus\limits_{z\in\alpha(i)} z.
\label{layers_and_z}
\end{align}

The proof of the Claim proceeds by induction on $j$.

The \textbf{induction base} is $j=0$.
For case \lequiv,
we have to consider $\varphi_0=t$ and $\varphi_1=\EF t$.
In layer $W_0$, all nodes only have itself as successor,
and therefore $t$ is satisfied in a node of layer $W_0$ if and only if $\EF t$ is satisfied by this node. 

For case \daneben, clearly $\KEFo,b_0\models t$ and $\KEFo,a_0\nmodels t$.

For case \ecircuit, $\KEFo,w_0\models t$ if and only if $\KEFo,w_0\in T$ if and only if $\apath_G(w_0,T)$.

For case \drueber, we proceed by induction on $i$. First, notice that for every $w_0$ in layer $0$,
$\KEFo,w_0\models\EF(\psi)$ iff $\KEFo,w_0\models \psi$, and $\KEFo,w_0\models\AG(\psi)$ iff $\KEFo,w_0\models\psi$.

The base case is $i=2$.
Every node in layer $W_0$ satisfies,
$\varphi_2= \AG(z_0 \oplus t \oplus \EF t)$.
For the inductive step, consider a node $w_0\in W_0$.
Notice that $\KEFo,w_0\nmodels\varphi_0\oplus\varphi_1$ (part \lequiv).
By the inductive hypothesis, 
all formulas $\varphi_q$ for even $q$ with $2\leq q <i$ are satisfied in $w_0$, 
and all formulas $\varphi_r$ for odd $r$ with $2\leq r <i$ are not satisfied in $w_0$.
By the semantics of $\oplus$ we can conclude
\begin{align}
\KEFo,w_0\models\bigoplus\limits_{l=0}^{i-1}\varphi_l & \text{~~~ if and only if~~~} \#\{m\in\{2,3,\ldots,i-1\}\mid m\text{ even}\} \text{ is odd.}
\label{bc1}
\end{align}

We have to consider the cases for odd resp.~even $i$ separately,
and we start with $i>2$ being odd. 
By the definition of $\alpha(i)$ we have
\begin{align}
z_0\in\alpha(i) \text{~~if and only if~~} \#\{m\in\{2,3,\ldots,i-1\}\mid m\text{ even}\}\text{ is odd.}
\label{bc2}
\end{align}
From (\ref{bc1}) and (\ref{bc2}) being equivalences with the same right-hand side, 
and applying (\ref{layers_and_z}) for $w_0$ in layer $j=0$,
we get
\begin{align*}
\KEFo,w_0 \models \bigoplus\limits_{z\in\alpha(i)} z & \text{~~~if and only if~~~} \KEFo,w_0\models\bigoplus\limits_{l=0}^{i-1} \varphi_l,
\end{align*}
and therefore
\begin{align*}
\KEFo,w_0\nmodels\bigoplus\limits_{z\in\alpha(i)} z ~\oplus~  \bigoplus\limits_{l=0}^{i-1}\varphi_l ,
& \text{~~yielding~~} \KEFo,w_0\nmodels\underbrace{\EF\Big(\bigoplus\limits_{z\in\alpha(i)} z ~\oplus~  \bigoplus\limits_{l=0}^{i-1}\varphi_l\Big)}_{=\varphi_i} .
\end{align*}

For even $i>2$, we have 
\begin{align}
z_0\in\alpha(i) \text{~~if and only if~~} \#\{ m\in\{2,3,\ldots,i-1\} \mid m \text{ even}\} \text{ is even.}
\label{bc3}
\end{align}
From (\ref{bc1}) and (\ref{bc3}), 
and applying (\ref{layers_and_z}) for $w_0$ in layer $j=0$,
we get 
\begin{align*}
\KEFo,w_0 \models \bigoplus\limits_{z\in\alpha(i)} z & \text{~~~if and only if~~~} \KEFo,w_0\nmodels \bigoplus\limits_{l=0}^{i-1} \varphi_l .
\end{align*}
Therefore, 
\begin{align*}
\KEFo,w_0\models\bigoplus\limits_{z\in\alpha(i)} z ~\oplus~  \bigoplus\limits_{l=0}^{i-1}\varphi_l ,
& \text{~~yielding~~} \KEFo,w_0\models\AG\Big(\bigoplus\limits_{z\in\alpha(i)} z ~\oplus~  \bigoplus\limits_{l=0}^{i-1}\varphi_l\Big) .
\end{align*}

This concludes the proofs of the base cases.

For the \textbf{induction step}, consider $j>0$.
We start with some essential observations for nodes $w_j\in W_j$.
For even $i>j$, the formula $\varphi_i$ has the form $\AG(\ldots)$,
and by the semantics of $\AG$,  we have that $\KEFo,w_j\models \varphi_{i}$ holds if and only if
\begin{align}
 & \KEFo, w_j\models \bigoplus\limits_{z\in\alpha(i)} z ~\oplus~ \bigoplus\limits_{l=0}^{i-1} \varphi_l, \text{~~~and} \label{obsef-1x}\\
 & \text{for all successors $v$ of $w_j$ holds } \KEFo, v\models \varphi_{i} . 
\end{align}
Since every successor $v$ of $w_j$ is in layer $j-1$, and $i\geq(j-1)+2$,
from part \drueber\ of the induction hypothesis follows that $\KEFo,v\models\varphi_i$ for all $v$ in layer $j-1$.
Therefore, $\KEFo,w_j\models \varphi_{i}$ is equivalent to (\ref{obsef-1x}).
Similarly, for odd $i>j$ holds $\KEFo,w_j\models \varphi_{i}$ if and only if
\begin{align*}
 & \KEFo,w_j\models \bigoplus\limits_{z\in\alpha(i)} z ~\oplus~ \bigoplus\limits_{l=0}^{i-1} \varphi_l, \text{~~~or} \\
 & \text{for some successor $v$ of $w_j$ holds } \KEFo,v\models \varphi_{i} .
\end{align*}
Since every successor $v$ of $w_j$ is in layer $j-1$, and $i\geq(j-1)+2$,
from part \drueber\ of the induction hypothesis follows that $\KEFo,v\nmodels\varphi_i$.
Therefore we obtain the first observation
\begin{align}
\text{for }i>j: ~~ \KEFo,w_j\models \varphi_{i} & \text{ if and only if } 
 \KEFo,w_j\models \bigoplus\limits_{z\in\alpha(i)} z ~\oplus~ \bigoplus\limits_{l=0}^{i-1} \varphi_l . \label{obsEFAG}
\end{align}

By the inductive hypothesis {\daneben} we know that
$\KEFo,b_r\models \varphi_r$ for all $r<j$.
For all odd $r<j$ it holds that $b_r$ is reachable from $w_j$.
Since for odd $r$, the formula $\varphi_r$ has the form $\EF(\ldots)$,
it follows that $\KEFo,w_j\models \varphi_r$ for all odd $r<j$.
The inductive hypothesis {\daneben} also yields that $\KEFo,a_q\nmodels \varphi_q$ for all even $q<j$.
Since all such $a_q$ are reachable from $w_j$, 
and for even $q$ the formula $\varphi_q$ has the form $\AG(\ldots)$,
it follows that $\KEFo,w_j\nmodels \varphi_q$ for all even $q<j$.
Both together yield for every $t<j$,
\begin{align}
\KEFo,w_j\models \bigoplus_{l=0}^{t} \varphi_l & \text{~~~ if and only if~~~ }
\#\{m\in\{0,1,2,\ldots,t\} \mid m\text{ odd}\} \text{ is odd.} \label{obsdrueber}
\end{align}

Now back to the inductive step.
We start with the \textbf{inductive step for part \lequiv}.
Let $w_j\in W_j$ for $j>0$.
By (\ref{obsEFAG}) we get
\begin{align}
\KEFo,w_j\models \varphi_{j+1} & \text{~~~if and only if~~~}\KEFo, w_j\models \bigoplus\limits_{z\in\alpha(j+1)} z ~\oplus~ \bigoplus\limits_{l=0}^{j} \varphi_l . \label{equiv-1}
\end{align}
Because $z_j\in\alpha(j+1)$ if and only if $\#\{m\in\{0,1,2,\ldots,j-1\}\mid m\text{ odd}\}$ is odd,
it follows with (\ref{layers_and_z}) and (\ref{obsdrueber}) that 
\begin{align*}
\KEFo,w_j \models \bigoplus\limits_{z\in\alpha(j+1)} z & \text{~~if and only if~~}\KEFo, w_j\models \bigoplus\limits_{l=0}^{j-1} \varphi_l
\end{align*}
and therefore
\begin{align*}
\KEFo,w_j \nmodels \bigoplus\limits_{z\in\alpha(j+1)} z ~\oplus~ \bigoplus\limits_{l=0}^{j-1} \varphi_l .
\end{align*}
Adding $\varphi_j$ we get 
\begin{align}
\KEFo,w_j\models \bigoplus\limits_{z\in\alpha(j+1)} z ~\oplus~ \underbrace{\bigoplus\limits_{l=0}^{j} \varphi_l}_{=(\bigoplus\limits_{l=0}^{j-1} \varphi_l)\oplus\varphi_j} & 
\text{~~~if and only if~~~} \KEFo,w_j\models \varphi_j. \label{equiv-5}
\end{align}
(\ref{equiv-1}) and (\ref{equiv-5}) yields $\KEFo,w_j\models \varphi_{j}$ if and only if $\KEFo,w_j\models\varphi_{j+1}$.

This also proves 
\begin{align}
\text{for all $j$:}~~ \KEFo,w_j\nmodels \varphi_j \oplus \varphi_{j+1} . \label{obsjj+1}
\end{align}

We continue with the \textbf{induction step for case \drueber} for $j>0$, and proceed by induction on $i$.
The base case is $i=j+2$.
By (\ref{obsEFAG}) we have 
\begin{align}
\KEFo,w_j\models \varphi_{j+2} & \text{~~if and only if~~} \KEFo,w_j\models \bigoplus\limits_{z\in\alpha(j+2)} z ~\oplus~ \bigoplus\limits_{l=0}^{j+1} \varphi_l . \label{drueber-1}
\end{align}

With (\ref{obsjj+1}) we get 
\begin{align}
\KEFo,w_j\models \bigoplus\limits_{z\in\alpha(j+2)} z ~\oplus~ \bigoplus\limits_{l=0}^{j+1} \varphi_l &
\text{~~iff~~} \KEFo,w_j\models \bigoplus\limits_{z\in\alpha(j+2)} z ~\oplus~ \bigoplus\limits_{l=0}^{j-1} \varphi_l. \label{drueber-2}
\end{align}
We consider the cases for even resp.~odd $j+2$ separately.
First, we consider odd $j+2$.
Since $\#\{ m\in\{j+2,\ldots,(j+2)-1\} \mid m \text{ even}\}=0$,
we get that $z_j\in\alpha(j+2)$ iff $\#\{ m \in\{0,1,2,\ldots,j-1\} \mid m \text{ odd} \}$ is odd.
With (\ref{obsdrueber}) we get 
\begin{align*}
\KEFo,w_j \models \bigoplus_{z\in\alpha(j+2)} z & \text{~~ if and only if ~~} w_j\models \bigoplus\limits_{l=0}^{j-1} \varphi_l 
\end{align*}
and thus
\begin{align}
\KEFo,w_j \nmodels \bigoplus_{z\in\alpha(j+2)} z ~\oplus~ \bigoplus\limits_{l=0}^{j-1} \varphi_l . \label{drueber-3}
\end{align}
From (\ref{drueber-1}), (\ref{drueber-2}), and (\ref{drueber-3}) we get $w_j \nmodels \varphi_{j+2}$ for odd $j$.

For even $j+2$, we proceed similarly.
Since $\#\{ m\in\{j+2,\ldots,(j+2)-1\} \mid m \text{ even}\}=0$,
we get that $z_j\in\alpha(j+2)$ iff $\#\{ m \in\{0,1,2,\ldots,j-1\} \mid m \text{ odd} \}$ is even.
With (\ref{obsdrueber}) we get 
\begin{align*}
\KEFo,w_j \models \bigoplus_{z\in\alpha(j+2)} z & \text{~~ if and only if ~~}\KEFo, w_j\nmodels \bigoplus\limits_{l=0}^{j-1} \varphi_l 
\end{align*}
and thus
\begin{align}
\KEFo,w_j \models \bigoplus_{z\in\alpha(j+2)} z ~\oplus~ \bigoplus\limits_{l=0}^{j-1} \varphi_l . \label{drueber-4}
\end{align}
From (\ref{drueber-1}), (\ref{drueber-2}), and (\ref{drueber-4}) we get $\KEFo,w_j \models \varphi_{j+2}$ for even $j$.

Now for the inductive step $i>j+2$.
From (\ref{obsEFAG}) we have
\begin{align}
\KEFo,w_j\models \varphi_{i} & \text{~~if and only if~~} \KEFo,w_j\models \bigoplus\limits_{z\in\alpha(i)} z ~\oplus~ \bigoplus\limits_{l=0}^{i-1} \varphi_l. \label{drueber-5}
\end{align}
By the inductive hypothesis, we know for $r<i$ holds
\begin{align*}
\KEFo,w_j\nmodels \varphi_r  \text{ for odd } r\geq j+2 \text{~~and~~} \KEFo, w_j\models \varphi_q \text{ for even } q\geq j+2 .
\end{align*}
This means
\begin{align}
\KEFo,w_j \models \bigoplus\limits_{l=j+2}^{i-1}\varphi_l & \text{~~if and only if~~} \#\{m\in\{j+2,\ldots,i-1\}\mid m\text{ even}\} \text{ is odd} .
\label{drueber-7}
\end{align}
With (\ref{obsdrueber}) we get
\begin{align*}
& \KEFo,w_j \models \bigoplus\limits_{l=0}^{j-1}\varphi_l \oplus \bigoplus\limits_{l=j+2}^{i-1}\varphi_l 
\text{~~iff~~} \\ & \#\{m\in\{0,1,\ldots,j-1\}\mid m\text{ odd}\}  + \#\{m\in\{j+2,\ldots,i-1\}\mid m\text{ even}\} \text{ is odd} .
\end{align*}
And with (\ref{obsjj+1}) $w_j\nmodels \varphi_j \oplus \varphi_{j+1}$ we eventually get
\begin{align}
\KEFo,w_j \models \bigoplus\limits_{l=0}^{i-1}\varphi_l 
\text{~~iff~~} & \#\{m\in\{0,1,\ldots,j-1\}\mid m\text{ odd}\} \notag \\ &  + \#\{m\in\{j+2,\ldots,i-1\}\mid m\text{ even}\} \text{ is odd} .
\label{drueber-8}
\end{align}

For odd $i$, we have
$z_j\in\alpha(i)$ iff $\#\{m\in\{0,1,\ldots,j-1\} \mid m \text{ odd}\}+\#\{m\in\{j+2,\ldots,i-1\} \mid m \text{ even}\}$ is odd.
With (\ref{drueber-8}) follows
\begin{align*}
\KEFo,w_j \nmodels \bigoplus\limits_{z\in\alpha(i)} z ~~\oplus~~  \bigoplus\limits_{l=0}^{i-1}\varphi_l \text{ ~~~~~(for odd $i$). }
\end{align*}
With (\ref{drueber-5}) follows $\KEFo,w_j\nmodels \varphi_i$ for odd $i>j+2$. 

For even $i$, we have
$z_j\in\alpha(i)$ iff $\#\{m\in\{0,1,2,\ldots,j-1\} \mid m \text{ odd}\}+\#\{m\in\{j+2,\ldots,i-1\} \mid m \text{ even}\}$ is even.
With (\ref{drueber-8}) follows
\begin{align*}
\KEFo,w_j \models \bigoplus\limits_{z\in\alpha(i)} z ~~\oplus~~  \bigoplus\limits_{l=0}^{i-1}\varphi_l \text{ ~~~~~(for even $i$). }
\end{align*}
With (\ref{drueber-5}) follows $\KEFo,w_j\models \varphi_i$ for even $i>j+2$. 
This concludes the proof of the inductive step for \drueber.

Now we consider the \textbf{inductive step for part \daneben}.
We start with even $i>0$ and the state $b_i$.
By the semantics of $\AG$ we get 
\begin{align}
\KEFo,b_i \models \varphi_i \text{~~iff~~}
\KEFo,b_i & \models \bigoplus\limits_{z\in\alpha(i)} z ~\oplus~ \bigoplus\limits_{l=0}^{i-1} \varphi_l \text{~~and~~} 
\KEFo,b_{i-1} \models \varphi_i . \label{bi-2} 
\end{align}

Part {\lequiv} of the general inductive hypothesis yields
that $\KEFo,b_{i-1} \models \varphi_{i}$ is equivalent to $\KEFo,b_{i-1} \models \varphi_{i-1}$,
and the latter holds due to the inductive hypothesis. 
Thus from (\ref{bi-2}) remains
\begin{align}
\KEFo,b_i \models \varphi_i & \text{~~if and only if~~} \KEFo,b_i \models \bigoplus\limits_{z\in\alpha(i)} z ~\oplus~ \bigoplus\limits_{l=0}^{i-1} \varphi_l .
\label{bi-3}
\end{align}
We now consider the right-hand side of (\ref{bi-3}).
With (\ref{obsdrueber}) we get
\begin{align}
\KEFo,b_i\models  \bigoplus\limits_{l=0}^{i-1} \varphi_l & \text{~~if and only if~~} \#\{m\in\{0,1,2,\ldots,i-1\} \mid m\text{ odd}\} \text{ is odd.}
\label{bi-4}
\end{align}
We have $z_{i}\in\alpha(i)$ iff $\#\{m\in\{0,1,2,\ldots,i-2\} \mid m \text{ odd}\}$ is odd.
Since $i-1$ is odd,
we get $z_{i}\in\alpha(i)$ iff $\#\{m\in\{0,1,2,\ldots,i-2,i-1\} \mid m \text{ odd}\}$ is even. 
With (\ref{bi-4}) we get 
\begin{align}
\KEFo,b_i\models \bigoplus\limits_{z\in\alpha(i)}z \text{~~iff~~}\KEFo, b_i\nmodels \bigoplus\limits_{l=0}^{i-1} \varphi_i,
\text{~~hence~~}
\KEFo,b_i \models \bigoplus\limits_{z\in\alpha(i)} z ~\oplus~ \bigoplus\limits_{l=0}^{i-1} \varphi_l .
\label{bi-5}
\end{align}
From (\ref{bi-5}) and (\ref{bi-3}) follows $\KEFo,b_i \models \varphi_i$ (for even $i$).

For even $i$ and state $a_i$, we have $\KEFo,a_i\models\varphi_i$ if and only if
\begin{align}
 \KEFo,a_i &  \models \bigoplus\limits_{z\in\alpha(i)} z ~\oplus~ \bigoplus\limits_{l=0}^{i-1} \varphi_l, \text{~~and} \notag \\
 \KEFo,b_{i-1} & \models \varphi_i, \text{~~and} \notag \\
 \KEFo,a_{i-1} & \models \varphi_i. \label{ai-A}
\end{align}

Since (\ref{ai-A}) is equivalent to $\KEFo,a_{i-1} \models \varphi_{i-1}$ (part {\lequiv} of the general inductive hypothesis),
and $\KEFo,a_{i-1} \nmodels \varphi_{i-1}$ (inductive hypothesis),
it follows that $\KEFo,a_i\nmodels\varphi_i$.

Now consider odd $i>0$ and state $b_i$.
By semantics of $\EF$ we get $\KEFo,b_i \models \varphi_i$ if and only if
\begin{align}
 & \KEFo,b_i \models \bigoplus\limits_{z\in\alpha(i)} z ~\oplus~ \bigoplus\limits_{l=0}^{i-1} \varphi_l, \text{~~or} \notag \\
 &\KEFo, b_{i-1} \models \varphi_i,  \text{~~or} \label{bi-B} \\
 &\KEFo, a_{i-1}\models \varphi_i. \notag
\end{align}
Part (\ref{bi-B}) follows from part {\lequiv} of the general inductive hypothesis
and the inductive hypothesis $\KEFo,b_{i-1} \models \varphi_{i-1}$.
Thus $\KEFo,b_i \models \varphi_i$ is proven.

For odd $i>0$ and state $a_i$, we have
$\KEFo,a_i \models \varphi_i$ if and only if
\begin{align*}
& \KEFo,a_i \models \bigoplus\limits_{z\in\alpha(i)} z ~\oplus~ \bigoplus\limits_{l=0}^{i-1} \varphi_l, \text{~~and} \\
& \KEFo,a_{i-1} \models \varphi_i.
\end{align*}
From the inductive hypothesis $\KEFo,a_{i-1}\nmodels \varphi_{i-1}$
and from part {\lequiv} of the general inductive hypothesis follows $\KEFo,a_{i-1} \nmodels \varphi_i$.
Thus
\begin{align}
\KEFo,a_i \models \varphi_i & \text{~~if and only if~~} \KEFo,a_i \models \bigoplus\limits_{z\in\alpha(i)} z ~\oplus~ \bigoplus\limits_{l=0}^{i-1} \varphi_l
\label{ai-3}
\end{align}
With (\ref{obsdrueber}) we have
\begin{align}
\KEFo,a_i\models  \bigoplus\limits_{l=0}^{i-1} \varphi_l \text{~~if and only if~~} \#\{m\in\{0,1,2,\ldots,i-1\} \mid m\text{ odd}\} \text{~is odd.}
\label{ai-3a}
\end{align}
But $z_{i}\in\alpha(i)$ iff $\#\{m\in\{0,1,2,\ldots,i-2\} \mid m \text{ odd}\}$ is odd.
Since $i-1$ is even,
we have $z_{i}\in\alpha(i)$ iff $\#\{m\in\{0,1,2,\ldots,i-2,i-1\} \mid m \text{ odd}\}$ is odd.
Thus 
\begin{align}
\KEFo,a_{i}\models \bigoplus\limits_{z\in\alpha(i)} z\text{~iff~}a_i\models  \bigoplus\limits_{l=0}^{i-1} \varphi_l, 
\text{~what yields~}
\KEFo,a_i\nmodels \bigoplus\limits_{z\in\alpha(i)} z ~\oplus~\bigoplus\limits_{l=0}^{i-1} \varphi_l.
\label{ai-4}
\end{align}
From (\ref{ai-4}) and (\ref{ai-3}) follows $\KEFo,a_i \nmodels \varphi_i$.

For the \textbf{inductive step of part {\ecircuit}}, let $w_i\in W_i\cap V$ be a node in layer $i>0$.
We start with even $i>0$ and formula $\varphi_i$ of the form $\AG(\ldots)$.
By the semantics of $\AG$ we have $\KEFo,w_i \models \varphi_i$ if and only if
\begin{align}
& \KEFo,w_i \models \bigoplus\limits_{z\in\alpha(i)} z ~\oplus~ \bigoplus\limits_{l=0}^{i-1} \varphi_l, \text{~~and} \label{wieven-1} \\
& \text{for all successors $v$ of $w_i$ holds } \KEFo,v \models \varphi_i . \label{wieven-2}
\end{align}
Since all successors $v$ of $w_i$ are in layer $i-1$,
with the general inductive hypothesis {\lequiv} we get that (\ref{wieven-2}) is equivalent to
\begin{align}
& \text{ for all successors $v$ of $w_i$ holds } \KEFo,v \models \varphi_{i-1} . \label{wieven-3a}
\end{align}
Since for the successor $b_{i-1}$ of $w_i$,
the general inductive hypothesis part {\daneben} yields $\KEFo,b_{i-1}\models\varphi_{i-1}$,
we get that (\ref{wieven-3a}) is equivalent to 
\begin{align}
& \text{ for all successors $v\in V\cap W_{i-1}$ of $w_i$ holds } \KEFo,v \models \varphi_{i-1} . \label{wieven-3}
\end{align}

Since $z_i\in\alpha(i)$ iff $\#\{m\in\{0,1,2,\ldots,i-2\}\mid m \text{ odd}\}$ is odd, and $i-1$ is odd,
we get $z_i\in\alpha(i)$ iff $\#\{m\in\{0,1,2,\ldots,i-2,i-1\}\mid m \text{ odd}\}$ is even.
With (\ref{obsdrueber}) follows
\begin{align*}
\KEFo,w_i \models \bigoplus\limits_{z\in\alpha(i)} z ~\oplus~ \bigoplus\limits_{l=0}^{i-1} \varphi_l .
\end{align*}
Thus (\ref{wieven-1}) holds, and with the equivalence of (\ref{wieven-2}) and (\ref{wieven-3}) we get 
\begin{align*}
\KEFo,w_i \models \varphi_i & \text{ ~~iff~~ for all successors $v\in V\cap W_{i-1}$ of $w_i$ } \KEFo,v \models \varphi_{i-1} . 
\end{align*}
By the inductive hypothesis and the construction of $\KEFo$ from $G$ we get
\begin{align*}
\KEFo,w_i \models \varphi_i & \text{ ~~iff~~ for all successors $v\in V\cap W_{i-1}$ of $w_i$ in $G$ } \apath_G(v,T) . 
\end{align*}
Since $i$ is even, $w_i$ is an $\forall$-node.
This yields what we look for, namely
\begin{align*}
\KEFo,w_i \models \varphi_i & \text{ ~~if and only if~~ } \apath_G(w_i,T) ~~~\text{(for even $i$)}. 
\end{align*}

Next we consider odd $i>0$.
Then $\varphi_i$ has the form $\EF(\ldots)$.
From the semantics of $\EF$ we get that $\KEFo,w_i \models \varphi_i$ if and only if
\begin{align}
& \KEFo,w_i \models \bigoplus\limits_{z\in\alpha(i)} z ~\oplus~ \bigoplus\limits_{l=0}^{i-1} \varphi_l, \text{~~or} \label{wiodd-1} \\
& \text{for some successor $v$ of $w_i$ holds } \KEFo,v \models \varphi_i . \label{wiodd-2}
\end{align}

Since $z_i\in\alpha(i)$ iff $\#\{m\in\{0,1,2,\ldots,i-2\}\mid m \text{ odd}\}$ is odd,
and $i-1$ is even, we get
$z_i\in\alpha(i)$ iff $\#\{m\in\{0,1,2,\ldots,i-2,i-1\}\mid m \text{ odd}\}$ is odd.
With (\ref{obsdrueber}) follows
\begin{align}
\KEFo,w_i \nmodels \bigoplus\limits_{z\in\alpha(i)}z ~\oplus \bigoplus\limits_{l=0}^{i-2} \varphi_l .
\label{wiodd-3}
\end{align}
what shows that (\ref{wiodd-1}) does not hold.
Thus from (\ref{wiodd-1}) and (\ref{wiodd-3}) we get
\begin{align*}
\KEFo,w_i \models \varphi_i & \text{~~iff~~ for some successor $v$ of $w_i$ holds } \KEFo,v \models \varphi_{i-1} . 
\end{align*}
For successor $a_{i-1}$ of $w_i$ holds $\KEFo,a_{i-1}\nmodels\varphi_{i}$ (general inductive hypothesis {\daneben}, {\lequiv}).
Therefore 
\begin{align*}
\KEFo,w_i \models \varphi_i & \text{~~iff~~ for some successor $v\in V\cap W_{i-1}$ of $w_i$ } \KEFo,v \models \varphi_{i-1} . 
\end{align*}
By the inductive hypothesis and the construction of $\KEFo$ from $G$ we get
\begin{align*}
\KEFo,w_i \models \varphi_i & \text{ ~~iff~~ for some successor $v\in V\cap W_{i-1}$ of $w_i$ in $G$ } \apath_G(v,T) . 
\end{align*}
Since $i$ is odd, $w_i$ is an $\exists$-node.
This concludes the proof of the Claim with
\begin{align*}
\KEFo,w_i \models \varphi_i & \text{ ~~if and only if~~ } \apath_G(w_i,T) ~~~\text{(for odd $i$)}. 
\end{align*}
\claimqed
\end{proof}

We now have that $\langle G,s,T \rangle \in \ASAGAP_{\log}$ if and only if $\langle \KEFo, s, \varphi_{\ell}\rangle\in\CTLMC{\EF,\oplus}$.
In order to estimate the size $|\varphi_i|$ of formula $\varphi_i$,
let $|\varphi_i|$ be the number of appearances of atoms in $\varphi$.
Then $|\varphi_0|=1$, and $|\varphi_{i+1}|\leq (i+1) + \sum_{j=0}^{i-1} |\varphi_j|$.
This yields $|\varphi_{i+1}|\leq 2\cdot|\varphi_i|+1$. 
Since the depth $\ell$ of $G$ is logarithmic in the size of $G$,
we get that $\varphi_l$ has size polynomial in the size of $G$.
Thus the reduction function described above can be computed in logarithmic space.
Since $\ASAGAP_{\log}$ is $\AC1$-complete,
it follows that $\CTLMC{\EF,\oplus}$ is $\AC1$-hard under logspace reducibility.
\end{proof}

\vspace*{2ex}